\documentclass[11pt]{amsart}

\usepackage[T1]{fontenc}
\usepackage[utf8]{inputenc}
\usepackage{lmodern}
\usepackage{microtype}
\usepackage{amsmath,amssymb,amsfonts,amsthm,mathtools}
\usepackage{enumitem}
\usepackage[colorlinks=true,linkcolor=blue,citecolor=blue,urlcolor=blue]{hyperref}
\hypersetup{
  pdftitle={Growing Alphabets in Canonical Shuffle Experiments: Likelihood-Ratio Laws, Estimation Bounds, and Low-Budget Equivariant Design},
  pdfauthor={Alex Shvets}
}

\setlist[enumerate]{label=(\roman*),itemsep=0.3em}
\setlist[itemize]{itemsep=0.3em}

\numberwithin{equation}{section}

\newtheorem{theorem}{Theorem}[section]
\newtheorem{lemma}[theorem]{Lemma}
\newtheorem{proposition}[theorem]{Proposition}
\newtheorem{corollary}[theorem]{Corollary}
\theoremstyle{definition}
\newtheorem{definition}[theorem]{Definition}
\newtheorem{remark}[theorem]{Remark}
\newtheorem{example}[theorem]{Example}

\newcommand{\R}{\mathbb R}
\newcommand{\N}{\mathbb N}
\newcommand{\E}{\mathbb E}
\newcommand{\PP}{\mathbb P}
\newcommand{\one}{\mathbf 1}
\newcommand{\dd}{\,\mathrm d}
\newcommand{\iid}{\stackrel{\mathrm{i.i.d.}}{\sim}}
\newcommand{\law}{\mathcal L}
\newcommand{\TV}{\mathrm{TV}}
\newcommand{\KL}{\mathrm{KL}}
\newcommand{\LeCam}{\Delta}
\newcommand{\GDP}{\mathrm{GDP}}

\newcommand{\Bin}{\mathrm{Bin}}
\newcommand{\Mult}{\mathrm{Mult}}
\newcommand{\Normal}{\mathcal N}

\DeclareMathOperator{\Var}{Var}
\DeclareMathOperator{\Cov}{Cov}
\DeclareMathOperator{\diag}{diag}

\title[Growing Alphabets and Canonical Shuffle Experiments]{Growing Alphabets in Canonical Shuffle Experiments:\\
Likelihood-Ratio Laws, Estimation Bounds, and Low-Budget Equivariant Design}

\author{Alex Shvets}
\address{Independent Researcher, Haifa, Israel}
\email{alt178332@gmail.com}

\date{May 2026}


\begin{document}

\begin{abstract}
We study the canonical one-step neighboring shuffle experiment generated by finite-output
$\varepsilon_0$-locally differentially private $d$-ary channels along growing alphabets, and the
associated finite-output frequency-estimation and mechanism-design problems under a canonical
pairwise $\chi^2$-budget.  The central object is the pairwise likelihood-ratio law $\nu_{ab,d}$,
the pushforward under the null row of the row ratio $W_d(\cdot\mid b)/W_d(\cdot\mid a)$.  We prove
an exact compression theorem: for the canonical homogeneous neighboring pair, the shuffled
histogram experiment is exactly equivalent at every finite $n$ to the quotient multinomial
experiment generated by $\nu_{ab,d}$.  Consequently, raw alphabet size is not the governing
invariant; alphabet growth improves canonical shuffled privacy precisely through collapse of the
worst pairwise likelihood-ratio law to $\delta_1$.  We establish the sharp pure-LDP endpoint
principle that the pairwise $\chi^2$ between two rows is at most $(e^{\varepsilon_0}-1)^2/e^{\varepsilon_0}$,
with equality if and only if the law is the endpoint two-point law, and we construct full-support
interior obstruction families on every even alphabet whose canonical shuffled privacy curve agrees
exactly with binary randomized response for all $d$.  We then prove a sharp diluting/persistent
dichotomy and give explicit finite-$n$ canonical hockey-stick bounds in the diluting case.
On the statistical side, the worst-case pairwise $\chi^2$-budget $\chi_\ast(W)$ governs an Assouad
lower bound for arbitrary estimators of the input frequency vector in the i.i.d.\ multinomial
model from the shuffled histogram, in a global two-regime form.  A symmetrization theorem reduces
the uniform-point Fisher criterion to permutation-equivariant channels.  Under a canonical
pairwise $\chi^2$-budget, calibrated GRR is not optimal; the optimal mechanism within the
permutation-equivariant low-budget regime is augmented GRR, which concentrates an aggressive local
signal on a random informative subset of users.  All privacy statements are about the canonical
homogeneous neighboring experiment, not the full shuffled-DP profile over arbitrary backgrounds.
\end{abstract}

\maketitle

\section{Introduction}

Let $n \ge 1$ users hold private inputs $x_1,\dots,x_n$.  In the shuffle model each user applies a local
randomizer and a trusted shuffler releases the multiset of local messages, equivalently the released histogram.
This model, initiated by the Encode--Shuffle--Analyze architecture of Bittau et al.\ \cite{BittauEtAl2017},
sits between local and central differential privacy and yields privacy amplification by anonymity
\cite{CheuEtAl2019,ErlingssonEtAl2019,BalleEtAl2019,FeldmanMcMillanTalwar2021,FeldmanMcMillanTalwar2023,TakagiLiew2026}.

For fixed finite output alphabets the asymptotic structure of neighboring shuffle experiments is now rather
well understood.  Part I of the trilogy \cite{ShvetsPartI} develops a sharp Gaussian/LAN/GDP theory in the
full-support regime; Part II \cite{ShvetsPartII} identifies critical Poisson, Skellam, and compound-Poisson
limits when the local randomizer itself becomes increasingly concentrated; Part III \cite{ShvetsPartIII}
describes the finite-alphabet dominant-block quotient geometry and its hybrid Gaussian--compound-Poisson
limit.  The present paper addresses a different question, left open in the trilogy and particularly pointed in
the discussion of Part III: what happens when the output alphabet grows?

A first guess is that more output symbols should automatically help privacy.  This is indeed what happens
for $d$-ary generalized randomized response (GRR): if one input changes from $a$ to $b$, then only the
two output coordinates $a$ and $b$ carry nontrivial pairwise likelihood-ratio values, while the remaining
$d-2$ coordinates are pairwise neutral.  Consequently the pairwise $\chi^2$ divergence is of order $d^{-1}$,
and the corresponding canonical shuffled privacy curve collapses as $d \to \infty$.
This $d$-dependent improvement for $k$-ary randomized response was first established theoretically by
Feldman, McMillan, and Talwar \cite[Corollary~4.2]{FeldmanMcMillanTalwar2021},
who proved a blanket privacy bound of order $e^{\varepsilon_0}/\!\sqrt{kn}$ for $k$-RR.
Our Proposition~\ref{prop:grr} gives the exact experiment-level counterpart.

The main message of the present paper is twofold.  First, the improvement of GRR is mechanism-specific, not universal.  There exist explicit
$\varepsilon_0$-LDP channels on $d$ symbols for which the pairwise $\chi^2$ divergence stays constant as
$d \to \infty$, and for which the exact canonical shuffled privacy curve of a worst-case neighboring pair is
\emph{identical} to the binary randomized-response privacy curve for every $d$.
Thus growing alphabets do not automatically amplify canonical shuffled privacy.
Second, the shuffle model has its own optimal mechanism geometry, fundamentally different from local DP.
The locally optimal mechanism (GRR at matched budget) is \emph{not} the low-budget equivariant optimum; the
low-budget equivariant optimum
mechanism concentrates an aggressive local signal on a random fraction of users.
In the low-budget regime this thinned augmented-GRR construction is, in fact, optimal among all
permutation-equivariant channels for the projected unbiased inverse-estimator risk, up to conditionally identical output refinements.

The correct invariant is not the raw alphabet size $d$, nor the ambient rank of the channel matrix, nor the
number of output coordinates on which rows differ.  It is the pushforward law of the pairwise likelihood ratio
\[
w_{ab,d}(y):=\frac{W_d(y\mid b)}{W_d(y\mid a)}
\]
under the null row $W_d(\cdot\mid a)$.  This law lives on the fixed compact interval
$[e^{-\varepsilon_0},e^{\varepsilon_0}]$ and completely determines the canonical one-step neighboring
shuffle experiment.

Thus the paper should be read less as a statement about alphabet size itself and more as an exact
invariance theory for one-step anonymous contamination.  Alphabet growth improves canonical shuffled
privacy precisely to the extent that the worst pairwise likelihood-ratio law loses mass away from $1$:
the diluting regime is exactly the collapse of these laws to $\delta_1$, while persistent
nondegenerate limits keep the experiment on the ordinary $n^{-1/2}$ shuffled Gaussian/GDP scale.
The scalar pairwise divergence
\[
I_{0,d}(a,b)=\int (r-1)^2\,\nu_{ab,d}(\mathrm dr)
\]
is the second moment of this exact invariant, and its worst-case version
\[
\chi_\ast(W):=\max_{a\neq b}\chi^2\!\bigl(W(\cdot\mid b)\,\|\,W(\cdot\mid a)\bigr)
\]
is the governing finite-dimensional obstruction.  On the privacy side, $\chi_\ast$ distinguishes
diluting from persistent growing-alphabet families and gives explicit finite-$n$ hockey-stick
bounds in the diluting case.  On the estimation side, the same quantity enters a universal
Assouad lower bound for frequency estimation, valid for arbitrary finite-output channels and
arbitrary estimators.  On the design side, taking $\chi_\ast$ as the canonical shuffle budget
leads to a different optimum from the locally calibrated GRR mechanism: in the low-budget symmetric
regime, the optimal channel is augmented GRR, which concentrates an aggressive local signal on a
random informative subset of messages.  Thus privacy obstruction, statistical hardness, and
low-budget equivariant design are governed by the same carrier: the pairwise likelihood-ratio
law, with $\chi_\ast$ as its canonical second-moment summary.

We now summarize the main results, organized into three groups: privacy structure, estimation bounds, and mechanism design.

\subsection*{I. Privacy structure}

\smallskip
\noindent
\textbf{(A) Exact likelihood-ratio quotient compression} (Theorem~\ref{thm:compression}).
For the canonical neighboring pair in which one input changes from $a$ to $b$,
the full histogram experiment is exactly equivalent to the multinomial experiment on the
likelihood-ratio level sets of $w_{ab,d}$.  The proof is exact and finite-$n$.

\smallskip
\noindent
\textbf{(B) Universal extremal $\chi^2$ bound} (Theorem~\ref{thm:chi2-bound}).
For every $\varepsilon_0$-LDP channel and every pair $a \neq b$,
\[
\chi^2\!\bigl(W_d(\cdot\mid b)\,\|\,W_d(\cdot\mid a)\bigr)
\le \frac{(e^{\varepsilon_0}-1)^2}{e^{\varepsilon_0}}.
\]
The bound is sharp, and equality holds precisely for endpoint two-point likelihood-ratio laws.

\smallskip
\noindent
\textbf{(C) Explicit obstruction family} (Theorem~\ref{thm:half-block}).
For every even $d$ we construct a full-support cyclic half-block channel $W_d$ for which,
along an opposite pair of inputs, the likelihood-ratio law is the same fixed two-point law for all $d$.
Hence the exact canonical shuffled privacy curve is exactly the binary randomized-response curve for all $d$.

\smallskip
\noindent
\textbf{(D) Sharp geometric dichotomy} (Theorem~\ref{thm:dichotomy}).
A channel family $\{W_d\}$ is \emph{diluting} iff its worst-case pairwise likelihood-ratio law collapses to
$\delta_1$, equivalently iff its worst-case pairwise $\chi^2$ divergence tends to zero.
It is \emph{persistent} iff some pairwise likelihood-ratio law has a nondegenerate subsequential limit,
equivalently iff the worst-case pairwise $\chi^2$ has positive limsup.
Persistent families remain on the ordinary shuffled Gaussian/GDP scale $n^{-1/2}$, with no extra
$d$-gain; diluting families yield asymptotically perfect fixed-$\varepsilon$ privacy.

\subsection*{II. Estimation bounds}

\smallskip
\noindent
\textbf{(E) Universal estimation lower bound}
(Theorems~\ref{thm:near-vertex-cr} and~\ref{thm:assouad-risk}).
For any channel $W:[d]\to\Delta(\mathcal Y)$ with arbitrary output alphabet and any estimator
$\widehat\theta$ of the input frequency vector from the shuffled histogram,
the minimax $\ell_2^2$ risk is at least of order $(d-1)/(n\chi_\ast(W))$
in the regime $n\chi_\ast(W)\gtrsim d^2$.
We give two independent proofs: a near-vertex Cram\'er--Rao argument for locally unbiased
estimators, and an Assouad argument for arbitrary estimators.
The quantity $\chi_\ast(W)=\max_{a\neq b}\chi^2(W(\cdot\mid b)\,\|\,W(\cdot\mid a))$ is therefore
not only the canonical privacy invariant but also a universal statistical obstruction.

\smallskip
\noindent
\textbf{(F) Symmetrization} (Theorem~\ref{thm:symmetrization}).
Averaging any channel over all input permutations does not increase the pairwise $\chi^2$ budget
and equalizes the Fisher eigenvalues.  Hence, for the uniform-point Fisher criterion under a budget
upper bound $\chi_\ast(W)\le C$, one may
restrict without loss of generality to permutation-equivariant channels.

\subsection*{III. Mechanism design}

\smallskip
\noindent
\textbf{(G) GRR is not universally optimal; the thinning principle and low-budget symmetric optimality}
(Theorem~\ref{thm:grr-mixtures}, Corollary~\ref{cor:aug-better}, Theorem~\ref{thm:equivariant-opt}).
Within the natural-orbit class of GRR blocks with null refinements, the optimal mechanism at canonical
shuffle budget $C$ for the projected unbiased inverse-estimator risk is not the calibrated GRR channel.
For $0<C\le C_\ast(d)$, the unique optimum in this class is an augmented GRR channel:
a fraction $p=C/C_\ast(d)$ of users applies GRR with $\lambda_\ast=\sqrt{d-1}$,
while the remaining fraction sends a common null symbol.
Thus the shuffle model rewards concentrating the local signal on a random informative subset of
messages rather than spreading noise uniformly.
Theorem~\ref{thm:equivariant-opt} shows that the same low-budget augmented-GRR construction is
optimal among all permutation-equivariant channels with arbitrary finite output alphabet, up to
conditionally identical output refinements.

\smallskip
\noindent
\textbf{(H) GRR is the unique optimizer within subset selection}
(Theorem~\ref{thm:ss-opt}).
Among all subset-selection mechanisms $SS(d,s,\lambda)$ calibrated to a common canonical $\chi^2$-budget,
the matched-budget risk is strictly increasing in the subset size $s$.
Hence GRR ($s=1$) is the unique optimizer in the entire subset-selection family.

The paper is nearly self-contained.  We reproduce the canonical exact likelihood-ratio identity from
Part I for convenience.  The main external probabilistic tools are the classical Berry--Esseen
inequality for i.i.d.\ sums and sequential compactness of probability measures on compact intervals.
The estimation lower bounds and the symmetric low-budget mechanism-design theorem are proved in full.

\paragraph{Scope and limitations.}
Throughout this paper, ``privacy structure'' refers to the canonical homogeneous neighboring
experiment in which one input changes from $a$ to $b$ while the remaining $n-1$ inputs are held
fixed at the common value $a$.  This is an exact, finite-$n$ object whose likelihood ratio depends
only on the pairwise likelihood-ratio law.  Full shuffled-$(\varepsilon,\delta)$-DP requires a
supremum over arbitrary neighboring datasets with arbitrary backgrounds for the unchanged users;
the canonical pair is one such pair, but not all such pairs reduce to it.  The compression
theorem, the universal $\chi^2$ endpoint principle, the half-block obstruction, the
diluting/persistent dichotomy, and the finite-$n$ certificate of
Corollary~\ref{cor:finite-n-cert} are statements about the canonical homogeneous experiment and
its directed hockey-stick profiles, not about the full shuffled-DP profile over arbitrary
backgrounds.  The estimation lower bounds (Sections~\ref{sec:estimation-lower}--\ref{sec:symmetrization})
are stated in the i.i.d.\ multinomial frequency-estimation model with the released shuffled
histogram as the observed statistic.  The mechanism-design results
(Sections~\ref{sec:thinning}--\ref{sec:ss}) are stated in the exact fixed-composition model and
concern the projected unbiased inverse-estimator risk under a pairwise $\chi_\ast$-budget.
Extending the privacy structure to arbitrary heterogeneous backgrounds is a natural direction left
open here.

\paragraph{Separation from prior work by the author.}
The present paper uses the canonical likelihood-ratio notation developed in earlier work, but its
main results are not consequences of the fixed-alphabet limit theory in that work.  The following
table records the separation of objects and regimes.

\begin{center}
\footnotesize
\renewcommand{\arraystretch}{1.18}
\begin{tabular}{p{0.16\textwidth}p{0.20\textwidth}p{0.25\textwidth}p{0.27\textwidth}}
\hline
Work & Central object / regime & What it proves & What is not addressed there \\
\hline
Part~I~\cite{ShvetsPartI}
&
Binary input, finite output, full support; Gaussian/LAN/GDP regime; also multi-message extensions.
&
Exact canonical LR identity, LAN/Gaussian shuffle limits, GDP privacy curves, and fixed-alphabet
multi-message asymptotics.
&
No growing-alphabet limit, no pure-LDP $\chi^2$ endpoint extremality, no full-support obstruction
families, no diluting/persistent dichotomy, no Assouad frequency-estimation lower bound, and no
mechanism design under a $\chi_\ast$-budget.
\\
\hline
Part~II~\cite{ShvetsPartII}
&
Finite alphabet with local randomizers entering singular or high-privacy regimes, including
$\varepsilon_0\to\infty$ scalings.
&
Non-Gaussian shuffle limits such as Poisson, Skellam, and compound-Poisson limits in special
critical regimes.
&
No growing-$d$ obstruction theory, no finite-$d$ pure-LDP endpoint principle, and no
frequency-estimation or mechanism-design lower bounds.
\\
\hline
Part~III~\cite{ShvetsPartIII}
&
Finite-alphabet dominant-block quotient geometry and boundary regimes.
&
Dominant-block reductions and hybrid Gaussian--compound-Poisson limits for fixed finite alphabets.
&
No growing-alphabet behavior, no exact pure-LDP $\chi^2$ extremality, no explicit full-support
persistent obstruction family indexed by $d$, and no augmented-GRR optimality under
$\chi_\ast$-budget.
\\
\hline
Anchored LR Geometry~\cite{ShvetsAnchoredLR}
&
Anchored simplex law as a common gauge for finite-output channels.
&
Universal extremality of binary randomized response, a rigidity converse, a trace-cap/two-orbit
$\chi_\ast$-frontier, and low-budget augmented-RR design in the anchored-law framework.
&
No growing-$d$ LR-law dilution/persistence theory, no explicit full-support persistent obstruction
family indexed by $d$, no finite-$n$ canonical dilution certificate, and no Assouad lower bound for
growing-alphabet multinomial frequency estimation.
\\
\hline
This paper
&
Growing $d$-ary finite-output pure-LDP channels; canonical homogeneous neighboring shuffle
experiments; frequency estimation and equivariant design under pairwise $\chi_\ast$-budget.
&
Exact LR-quotient compression, sharp pure-LDP $\chi^2$ endpoint principle, full-support obstruction
families, diluting/persistent dichotomy, finite-$n$ canonical certificates, Assouad lower bounds,
symmetrization, augmented-GRR low-budget optimality, and subset-selection comparison.
&
Does not claim a full shuffled-DP profile over arbitrary heterogeneous backgrounds, does not treat
continuous non-pure-LDP randomizers, and does not solve the global permutation-equivariant frontier
above the low-budget threshold $C_\ast(d)$.
\\
\hline
\end{tabular}
\end{center}

\paragraph{Relation to Takagi and Liew.}
Takagi and Liew~\cite{TakagiLiew2026,TakagiLiew2604} have developed a complementary shuffle-index
framework: their PODS work~\cite{TakagiLiew2026} introduces a scalar shuffle index for a regular
class of local randomizers (including Gaussian-type non-pure-LDP mechanisms) and establishes
asymptotic privacy-profile bounds beyond pure local DP, while their subsequent
preprint~\cite{TakagiLiew2604} studies single-message shuffle optimization for unbiased
$d$-dimensional vector mean estimation over $\mathbb B_2^d$ and proves a high-privacy minimax
theory with normalized risk scale $d\chi^2$, attained asymptotically by a blanket-mixed Gaussian
mechanism.  The present paper treats a different finite-output pure-LDP problem: $d$-ary
$\varepsilon_0$-LDP channels under growing alphabets, exact LR-quotient compression, and
frequency-estimation/design under the pairwise canonical budget $\chi_\ast(W)$.  The two
formulations optimize different objects: the shuffle index is a privacy/noise scale on the local
randomizer, while $\chi_\ast(W)$ is a pairwise signal/divergence budget on the channel rows.
Their shuffle-index lower bound and the present paper's $\chi_\ast$-based lower bound are
suggestive of a reciprocal signal-to-noise scaling, but the two invariants are not formally
interchangeable, and the two formulations address different primary regimes and canonical
objects: their shuffle-index framework covers a broad class of regular randomizers, including
non-pure-LDP Gaussian-type mechanisms and some finite-output mechanisms such as $k$-RR, while
the present paper develops a finite-output pure-LDP row-ratio geometry for growing $d$, canonical
neighboring pairs, and $\chi_\ast$-budgeted frequency estimation and design.  Beyond the
canonical-experiment scope addressed here, the present paper contributes the sharp pure-LDP
pairwise $\chi^2$ extremal bound, explicit persistent obstruction families and the
diluting/persistent dichotomy, finite-output symmetrization, and augmented-GRR low-budget
optimality among permutation-equivariant channels; these questions are outside the scope of the
cited works.

\section{Model and notation}

Fix $d \ge 2$ and $n \ge 1$, and write $[d]=\{1,\dots,d\}$.
A local randomizer is a Markov kernel
\[
W_d : [d] \to \Delta([d]),
\qquad
x \mapsto W_d(\cdot\mid x).
\]
The released shuffle transcript is the histogram
\[
N=(N_1,\dots,N_d), \qquad \sum_{y=1}^d N_y = n,
\]
obtained by drawing independent local outputs
$Y_i \sim W_d(\cdot\mid x_i)$ and then forgetting the order.

We write
\[
\Delta_d:=\Bigl\{\theta\in[0,1]^d:\sum_{x=1}^d \theta_x=1\Bigr\}
\]
for the probability simplex on \([d]\).

\begin{definition}
The channel $W_d$ is \emph{$\varepsilon_0$-locally differentially private} if
\[
\max_{x,x' \in [d]} \max_{y \in [d]}
\frac{W_d(y\mid x)}{W_d(y\mid x')} \le e^{\varepsilon_0},
\]
with the convention that $0/0=1$ and $p/0=+\infty$ for $p>0$.
\end{definition}

\begin{lemma}[Common support under pure LDP]\label{lem:common-support}
Let $W_d$ be $\varepsilon_0$-LDP with $\varepsilon_0<\infty$.
For each output symbol $y \in [d]$, either $W_d(y\mid x)=0$ for every input $x$, or
$W_d(y\mid x)>0$ for every input $x$.
\end{lemma}

\begin{proof}
If $W_d(y\mid x)=0$ for some $x$ and $W_d(y\mid x')>0$ for some $x'$, then
\[
\frac{W_d(y\mid x')}{W_d(y\mid x)} = +\infty,
\]
contradicting finite $\varepsilon_0$-LDP.
\end{proof}

Henceforth we discard outputs of common zero mass and identify the output alphabet with this common support,
so that all rowwise likelihood ratios are finite.
Throughout the paper we assume \(\varepsilon_0>0\); the case $\varepsilon_0=0$ forces $W_d(\cdot\mid x)=W_d(\cdot\mid x')$ for all $x,x'$, giving $I_{0,d}(a,b)=0$ and trivial privacy for every $d$.

\subsection*{Canonical neighboring pair}

Fix distinct inputs $a,b \in [d]$.
The \emph{canonical pair} is the binary experiment
\[
P_{n,d}^{ab} \quad \text{versus} \quad Q_{n,d}^{ab},
\]
where under $P_{n,d}^{ab}$ all $n$ users hold input $a$, while under $Q_{n,d}^{ab}$ one distinguished user
holds input $b$ and the remaining $n-1$ users hold input $a$.
The released statistic is always the histogram $N$.

Write
\[
P_d^{a}:=W_d(\cdot\mid a), \qquad P_d^{b}:=W_d(\cdot\mid b),
\]
and define the pairwise likelihood-ratio function
\[
w_{ab,d}(y):=\frac{P_d^{b}(y)}{P_d^{a}(y)},
\qquad y \in [d].
\]
By $\varepsilon_0$-LDP,
\[
e^{-\varepsilon_0} \le w_{ab,d}(y)\le e^{\varepsilon_0}
\qquad\forall y.
\]

We also define the pairwise chi-square divergence
\[
I_{0,d}(a,b)
:=
\chi^2\!\bigl(P_d^{b}\,\|\,P_d^{a}\bigr)
=
\sum_{y=1}^d \frac{(P_d^{b}(y)-P_d^{a}(y))^2}{P_d^{a}(y)}
=
\E_{P_d^a}\!\bigl[(w_{ab,d}(Y)-1)^2\bigr].
\]
Equivalently,
\[
I_{0,d}(a,b)=\Var_{P_d^a}(w_{ab,d}(Y)),
\]
because $\E_{P_d^a}[w_{ab,d}(Y)]=1$.
In the notation of Part I, $I_{0,d}(a,b)$ is the canonical Fisher constant
$v^\top \Sigma^{+}v$ for the pair $(a,b)$; see \cite[Sec.~2.4]{ShvetsPartI}.

Finally, define the worst-case pairwise divergence
\[
I_d^\star := \max_{a\neq b} I_{0,d}(a,b).
\]

\subsection*{Privacy curves}

For a dominated binary experiment $Q \ll P$ with likelihood ratio $L=\dd Q/\dd P$,
the one-sided privacy curve is
\[
\delta_{Q\|P}(\varepsilon)
:=
\sup_A \{Q(A)-e^\varepsilon P(A)\}
=
\E_P\bigl[(L-e^\varepsilon)_+\bigr],
\qquad \varepsilon\ge0.
\]
Similarly,
\[
\delta_{P\|Q}(\varepsilon)=\E_P\bigl[(1-e^\varepsilon L)_+\bigr]
=\E_Q\bigl[(L^{-1}-e^\varepsilon)_+\bigr].
\]
The two-sided privacy curve is
\[
\delta_{P,Q}(\varepsilon):=\max\{\delta_{Q\|P}(\varepsilon),\delta_{P\|Q}(\varepsilon)\}.
\]

For $\mu>0$ we write
\[
\delta_{\GDP(\mu)}(\varepsilon)
:=
\Phi\!\left(-\frac{\varepsilon}{\mu}+\frac{\mu}{2}\right)
-
e^\varepsilon
\Phi\!\left(-\frac{\varepsilon}{\mu}-\frac{\mu}{2}\right),
\]
the one-sided privacy curve of the Gaussian pair
\(
\Normal(-\mu^2/2,\mu^2)
\)
versus
\(
\Normal(\mu^2/2,\mu^2)
\).

\section{LR-quotient compression}

We begin by reproducing the exact canonical likelihood-ratio identity from Part~I; compare
\cite[Lemma~3.2]{ShvetsPartI}.

\begin{lemma}[Exact canonical likelihood ratio]\label{lem:canonical-LR}
Fix $a\neq b$ and let $L_{n,d}^{ab}$ be the likelihood ratio
\[
L_{n,d}^{ab}(N):=\frac{Q_{n,d}^{ab}(N)}{P_{n,d}^{ab}(N)}.
\]
Then
\[
L_{n,d}^{ab}(N)=\frac1n \sum_{y=1}^d N_y\, w_{ab,d}(y). \tag{3.1}\label{eq:canonical-lr}
\]
\end{lemma}

\begin{proof}
Write $p_y=P_d^a(y)$ and $q_y=P_d^b(y)$.
Under $P_{n,d}^{ab}$ the histogram law is
\[
P_{n,d}^{ab}(N)
=
\frac{n!}{\prod_{y=1}^d N_y!}\prod_{y=1}^d p_y^{N_y}.
\]
Under $Q_{n,d}^{ab}$ one user is distributed according to $(q_y)$ and the remaining $n-1$ users according
to $(p_y)$.  Conditioning on the output of the distinguished user gives
\[
Q_{n,d}^{ab}(N)
=
\sum_{z=1}^d q_z
\frac{(n-1)!}{(N_z-1)!\prod_{y\neq z}N_y!}
\,p_z^{N_z-1}\prod_{y\neq z}p_y^{N_y},
\]
where by convention the term is zero if $N_z=0$.
Dividing by $P_{n,d}^{ab}(N)$ yields
\[
\frac{Q_{n,d}^{ab}(N)}{P_{n,d}^{ab}(N)}
=
\sum_{z=1}^d \frac{q_z}{p_z}\frac{N_z}{n}
=
\frac1n\sum_{z=1}^d N_z\,w_{ab,d}(z).
\]
\end{proof}

\begin{theorem}[LR-quotient compression]\label{thm:compression}
Fix a pair $a\neq b$.
Let
\[
\mathcal R_{ab,d}:=\{w_{ab,d}(y):y\in[d]\}
\subset [e^{-\varepsilon_0},e^{\varepsilon_0}]
\]
be the set of distinct pairwise likelihood-ratio values, and for each $r\in\mathcal R_{ab,d}$ define the
level set
\[
B_r:=\{y\in[d]: w_{ab,d}(y)=r\}.
\]
Let
\[
M_r:=\sum_{y\in B_r}N_y,
\qquad
r\in\mathcal R_{ab,d},
\]
and let $M=(M_r)_{r\in\mathcal R_{ab,d}}$.

Then the following hold.
\begin{enumerate}
\item Under $P_{n,d}^{ab}$, the vector $M$ is multinomial:
\[
M \sim \Mult\!\bigl(n;(p_r)_{r\in\mathcal R_{ab,d}}\bigr),
\qquad
p_r:=P_d^a(B_r).
\]
\item The exact likelihood ratio is
\[
L_{n,d}^{ab}(N)=\frac1n\sum_{r\in\mathcal R_{ab,d}} r\,M_r. \tag{3.2}\label{eq:compressed-lr}
\]
\item The statistic $M$ is sufficient for the binary experiment
\(
\{P_{n,d}^{ab},Q_{n,d}^{ab}\}
\).
Equivalently, the full shuffled histogram experiment is exactly equivalent to the experiment on the
$\lvert \mathcal R_{ab,d}\rvert$-dimensional quotient histogram $M$.
\end{enumerate}
\end{theorem}

\begin{proof}
(i)  Under $P_{n,d}^{ab}$ each local output falls in block $B_r$ with probability
\(
p_r=P_d^a(B_r)
\),
independently across users, so the block counts are multinomial.

(ii)  Starting from \eqref{eq:canonical-lr} and grouping equal likelihood-ratio values,
\[
L_{n,d}^{ab}(N)
=
\frac1n \sum_{y=1}^d N_y\,w_{ab,d}(y)
=
\frac1n \sum_{r\in\mathcal R_{ab,d}} r\sum_{y\in B_r}N_y
=
\frac1n \sum_{r\in\mathcal R_{ab,d}} r\,M_r.
\]

(iii)  By (ii), the likelihood ratio is measurable with respect to $\sigma(M)$.
In a dominated binary experiment, this already implies sufficiency.  For completeness we verify the
conditional-law identity directly.
Fix a realization $m=(m_r)_{r\in\mathcal R_{ab,d}}$ of $M$.
Then for every histogram $N$ with $M(N)=m$,
\[
Q_{n,d}^{ab}(N)=L_{n,d}^{ab}(N)\,P_{n,d}^{ab}(N)
=
\left(\frac1n\sum_{r\in\mathcal R_{ab,d}}r\,m_r\right) P_{n,d}^{ab}(N).
\]
Summing over all histograms with the same block count vector $m$ gives
\[
Q_{n,d}^{ab}(M=m)
=
\left(\frac1n\sum_{r\in\mathcal R_{ab,d}}r\,m_r\right)P_{n,d}^{ab}(M=m),
\]
and therefore
\[
Q_{n,d}^{ab}(N\mid M=m)=P_{n,d}^{ab}(N\mid M=m).
\]
Thus the backward Markov kernel $m\mapsto \law_{P_{n,d}^{ab}}(N\mid M=m)$, defined arbitrarily on
realizations $m$ outside the support of $\law(M)$, reconstructs the original
histogram experiment from the quotient experiment.  This is exact experiment equivalence.
\end{proof}

\begin{corollary}[The correct invariant]\label{cor:correct-invariant}
For the canonical neighboring pair $(a,b)$ the exact experiment depends only on the pushforward law
\[
\nu_{ab,d}:=P_d^a \circ w_{ab,d}^{-1}
\]
on the compact interval $[e^{-\varepsilon_0},e^{\varepsilon_0}]$.
In particular, the raw alphabet size $d$ is not a complete invariant of the canonical shuffle experiment.
\end{corollary}

\begin{proof}
The quotient histogram law in Theorem~\ref{thm:compression} is the multinomial law generated by the
atom masses of $\nu_{ab,d}$, and the exact likelihood ratio is the linear functional
\(
m \mapsto n^{-1}\sum_r r m_r
\).
\end{proof}

\section{Universal chi-square bound and extremal channels}

We next prove the universal upper bound on pairwise $\chi^2$ divergence.
The proof is a one-line variance identity usually attributed to Bhatia and Davis
\cite{BhatiaDavis2000}; we include it for completeness.

\begin{lemma}[Bhatia--Davis variance bound]\label{lem:bhatia-davis}
Let $X$ be a real random variable satisfying
\[
m \le X \le M \quad \text{a.s.}
\]
and let $\mu=\E[X]$.
Then
\[
\Var(X)\le (M-\mu)(\mu-m). \tag{4.1}\label{eq:bd}
\]
Equality holds if and only if $X\in\{m,M\}$ almost surely.
\end{lemma}

\begin{proof}
Expand the nonnegative random variable $(M-X)(X-m)$:
\[
(M-X)(X-m)= -X^2 +(M+m)X - Mm.
\]
Taking expectations and using
\(
\E[X^2]=\Var(X)+\mu^2
\),
we obtain
\[
\begin{aligned}
\E[(M-X)(X-m)]
&=
-\Var(X)-\mu^2 +(M+m)\mu - Mm\\
&=
(M-\mu)(\mu-m)-\Var(X).
\end{aligned}
\]
Since the left-hand side is nonnegative, \eqref{eq:bd} follows.
Equality holds if and only if $(M-X)(X-m)=0$ almost surely, i.e.\ iff
$X\in\{m,M\}$ a.s.
\end{proof}

\begin{theorem}[Universal chi-square bound and extremizers]\label{thm:chi2-bound}
Let $W_d$ be an $\varepsilon_0$-LDP channel and write $\lambda=e^{\varepsilon_0}$.
For every pair $a\neq b$,
\[
0 \le I_{0,d}(a,b)
=
\chi^2\!\bigl(P_d^b\,\|\,P_d^a\bigr)
\le \frac{(\lambda-1)^2}{\lambda}. \tag{4.2}\label{eq:chi-bound}
\]
Equality holds if and only if, under $P_d^a$, the likelihood ratio
$w_{ab,d}(Y)$ takes only the two endpoint values
$\lambda^{-1}$ and $\lambda$ with probabilities
\[
P_d^a\!\bigl(w_{ab,d}(Y)=\lambda^{-1}\bigr)=\frac{\lambda}{1+\lambda},
\qquad
P_d^a\!\bigl(w_{ab,d}(Y)=\lambda\bigr)=\frac{1}{1+\lambda}. \tag{4.3}\label{eq:endpoint-masses}
\]
Equivalently, the output alphabet can be partitioned as
\(
[d]=B_-\sqcup B_+
\)
such that
\[
w_{ab,d}(y)=\lambda^{-1}\ \text{for } y\in B_-,
\qquad
w_{ab,d}(y)=\lambda\ \text{for } y\in B_+,
\]
with
\[
P_d^a(B_-)=\frac{\lambda}{1+\lambda},
\qquad
P_d^a(B_+)=\frac{1}{1+\lambda}. \tag{4.4}\label{eq:block-masses}
\]
\end{theorem}

\begin{proof}
Fix $a\neq b$ and set
\[
R:=w_{ab,d}(Y), \qquad Y\sim P_d^a.
\]
Because $W_d$ is $\varepsilon_0$-LDP,
\[
\lambda^{-1}\le R\le \lambda \qquad \text{a.s.}
\]
Also
\[
\E[R]=\sum_{y=1}^d P_d^a(y)\frac{P_d^b(y)}{P_d^a(y)}
=\sum_{y=1}^d P_d^b(y)=1.
\]
Finally,
\[
I_{0,d}(a,b)
=
\E[(R-1)^2]
=
\Var(R),
\]
again because $\E[R]=1$.

Applying Lemma~\ref{lem:bhatia-davis} with
\(
m=\lambda^{-1}
\),
\(
M=\lambda
\),
\(
\mu=1
\),
we obtain
\[
I_{0,d}(a,b)=\Var(R)\le (\lambda-1)(1-\lambda^{-1})
=\frac{(\lambda-1)^2}{\lambda}.
\]
This proves \eqref{eq:chi-bound}.

Equality in Lemma~\ref{lem:bhatia-davis} holds iff $R\in\{\lambda^{-1},\lambda\}$ almost surely.
Since $\E[R]=1$, if
\(
p=P_d^a(R=\lambda)
\),
then
\[
p\lambda +(1-p)\lambda^{-1}=1
\quad\Longrightarrow\quad
p=\frac{1}{1+\lambda}.
\]
Hence
\(
P_d^a(R=\lambda^{-1})=\lambda/(1+\lambda)
\),
which is exactly \eqref{eq:endpoint-masses}.  The equivalence with \eqref{eq:block-masses}
is immediate after setting
\(
B_+=\{y:w_{ab,d}(y)=\lambda\}
\)
and
\(
B_-=\{y:w_{ab,d}(y)=\lambda^{-1}\}
\).
\end{proof}

\begin{corollary}[Universal upper bound on the Gaussian scale]\label{cor:mu-bound}
For every pair $(a,b)$ and every $n\ge1$, the canonical Gaussian parameter
\[
\mu_{n,d}(a,b):=\sqrt{\frac{I_{0,d}(a,b)}{n}}
\]
satisfies
\[
\mu_{n,d}(a,b)\le \frac{\lambda-1}{\sqrt{\lambda\,n}}. \tag{4.5}\label{eq:mu-universal}
\]
In particular,
\[
\mu_{n,d}^\star:=\max_{a\neq b}\mu_{n,d}(a,b)
\le \frac{\lambda-1}{\sqrt{\lambda\,n}}.
\]
\end{corollary}

\begin{proof}
Immediate from Theorem~\ref{thm:chi2-bound}.
\end{proof}

\begin{remark}
Corollary~\ref{cor:mu-bound} is an upper bound on the \emph{canonical Gaussian/GDP scale}.
It does not assert finite-$n$ GDP for arbitrary channels.
Section~\ref{sec:dichotomy} proves that whenever
\( n\,I_{0,d}(a,b)\to\infty \),
the canonical neighboring experiment is asymptotically GDP with parameter
\(
\mu_{n,d}(a,b)=\sqrt{I_{0,d}(a,b)/n}
\).
The obstruction family below shows that \eqref{eq:mu-universal} is sharp.
\end{remark}

\section{An explicit obstruction family}

We now construct explicit $\varepsilon_0$-LDP channels on growing alphabets for which the exact shuffled
privacy curve is identical to the binary randomized-response curve for every $d$.
Throughout this section $\lambda=e^{\varepsilon_0}$.

\begin{definition}[Half-block cyclic channel]\label{def:half-block}
Assume $d$ is even and identify $[d]$ with the cyclic group $\mathbb Z_d$.
For each $x\in\mathbb Z_d$, define the half-block
\[
A_x:=\{x,x+1,\dots,x+d/2-1\}\pmod d.
\]
The channel $W_d$ is defined by
\[
W_d(y\mid x)
=
\frac{2\lambda}{d(1+\lambda)}\,\one\{y\in A_x\}
+
\frac{2}{d(1+\lambda)}\,\one\{y\notin A_x\}. \tag{5.1}\label{eq:half-block}
\]
\end{definition}

\begin{theorem}[Explicit obstruction family]\label{thm:half-block}
Let $W_d$ be the half-block channel of Definition~\ref{def:half-block}, and fix the opposite pair
\[
b=a+d/2 \pmod d.
\]
Then the following hold.

\begin{enumerate}
\item $W_d$ is $\varepsilon_0$-LDP for every even $d$.
\item
\[
I_{0,d}(a,b)=\chi^2\!\bigl(P_d^b\,\|\,P_d^a\bigr)=\frac{(\lambda-1)^2}{\lambda}, \tag{5.2}\label{eq:halfblock-chi}
\]
so the pairwise $\chi^2$ divergence is independent of $d$ and saturates the universal upper bound.
\item The LR quotient has exactly two atoms:
\[
\mathcal R_{ab,d}=\{\lambda^{-1},\lambda\}.
\]
The sufficient quotient statistic is
\[
K:=\sum_{y\notin A_a}N_y,
\]
and under $P_{n,d}^{ab}$,
\[
K\sim \Bin\!\left(n,\frac{1}{1+\lambda}\right). \tag{5.3}\label{eq:K-bin}
\]
\item The exact canonical likelihood ratio is
\[
L_{n,d}^{ab}(N)
=
\lambda^{-1}+\frac{K}{n}\bigl(\lambda-\lambda^{-1}\bigr). \tag{5.4}\label{eq:halfblock-LR}
\]
Hence the exact one-sided privacy curve is
\[
\delta_{Q_{n,d}^{ab}\|P_{n,d}^{ab}}(\varepsilon)
=
\sum_{k=0}^n \binom{n}{k} q^k(1-q)^{n-k}
\left(\lambda^{-1}+\frac{k}{n}\bigl(\lambda-\lambda^{-1}\bigr)-e^\varepsilon\right)_+, \tag{5.5}\label{eq:halfblock-delta}
\]
where \( q=(1+\lambda)^{-1} \).
This is exactly the binary randomized-response privacy curve.
\item For every sequence $n\to\infty$,
the canonical Gaussian/GDP scale is
\[
\mu_{n,d}(a,b)=\sqrt{\frac{I_{0,d}(a,b)}{n}}
=\frac{\lambda-1}{\sqrt{\lambda\,n}}, \tag{5.6}\label{eq:halfblock-mu}
\]
the same as for binary randomized response.
Thus growing $d$ gives no additional privacy amplification along this pair.
\end{enumerate}
\end{theorem}

\begin{proof}
(i)  Every entry of every row equals either
\(
2\lambda/[d(1+\lambda)]
\)
or
\(
2/[d(1+\lambda)]
\).
Hence for every $x,x',y$ the ratio
\(
W_d(y\mid x)/W_d(y\mid x')
\)
is one of
\(
1,\lambda,\lambda^{-1}
\),
so $W_d$ is $\varepsilon_0$-LDP.

(ii)  For the opposite pair $b=a+d/2$, the half-blocks are complementary:
\[
A_b=A_a^c.
\]
Therefore
\[
P_d^a(y)=
\frac{2\lambda}{d(1+\lambda)}\,\one\{y\in A_a\}
+
\frac{2}{d(1+\lambda)}\,\one\{y\notin A_a\},
\]
while
\[
P_d^b(y)=
\frac{2}{d(1+\lambda)}\,\one\{y\in A_a\}
+
\frac{2\lambda}{d(1+\lambda)}\,\one\{y\notin A_a\}.
\]
Hence the pairwise likelihood ratio is
\[
w_{ab,d}(y)
=
\lambda^{-1}\one\{y\in A_a\}
+
\lambda\one\{y\notin A_a\}. \tag{5.7}\label{eq:w-halfblock}
\]
Under $P_d^a$,
\[
P_d^a(A_a)=\frac{\lambda}{1+\lambda},
\qquad
P_d^a(A_a^c)=\frac{1}{1+\lambda}.
\]
Thus
\[
I_{0,d}(a,b)
=
\frac{\lambda}{1+\lambda}\,(\lambda^{-1}-1)^2
+
\frac{1}{1+\lambda}\,(\lambda-1)^2
=
\frac{(\lambda-1)^2}{\lambda}.
\]
This is \eqref{eq:halfblock-chi}.

(iii)  Equation \eqref{eq:w-halfblock} shows that the only likelihood-ratio levels are
$\lambda^{-1}$ and $\lambda$.
The quotient statistic can therefore be chosen as the count of outputs in the $\lambda$-block,
namely
\[
K=\sum_{y\notin A_a}N_y.
\]
Under $P_{n,d}^{ab}$ each user falls in $A_a^c$ with probability
\[
q=P_d^a(A_a^c)=\frac{1}{1+\lambda},
\]
so $K\sim \Bin(n,q)$.

(iv)  The exact quotient formula \eqref{eq:compressed-lr} from
Theorem~\ref{thm:compression} gives
\[
L_{n,d}^{ab}(N)
=
\frac{1}{n}\bigl[\lambda^{-1}(n-K)+\lambda K\bigr]
=
\lambda^{-1}+\frac{K}{n}(\lambda-\lambda^{-1}),
\]
which is \eqref{eq:halfblock-LR}.  Plugging this into the positive-part formula for
\(
\delta_{Q\|P}
\)
and using \eqref{eq:K-bin} yields \eqref{eq:halfblock-delta}.  This is exactly the classical binary
randomized-response expression.

(v)  Equation \eqref{eq:halfblock-mu} is immediate from \eqref{eq:halfblock-chi}.
The absence of additional $d$-gain is therefore exact at the Fisher/Gaussian scale, and
Theorem~\ref{thm:dichotomy}(iii) below shows that the canonical neighboring experiment remains on the
ordinary shuffled Gaussian/GDP scale.
\end{proof}

\begin{remark}
Theorem~\ref{thm:half-block} is a genuinely interior obstruction.
All rows have full support and no vanishing minority block is involved.
This is different from the strong-boundary obstruction of Part III
\cite[Prop.~6.1]{ShvetsPartIII}, where a bounded minority block can remain visible to exact histogram tests
while disappearing in the projected weak limit.
\end{remark}

\section{A sharp geometric dichotomy}\label{sec:dichotomy}

We now characterize exactly when growing alphabets do and do not force additional privacy amplification.

\begin{definition}
A family of $\varepsilon_0$-LDP channels $\{W_d\}_{d\ge2}$ is called
\begin{itemize}
\item \emph{diluting} if
\[
I_d^\star=\max_{a\neq b} I_{0,d}(a,b)\longrightarrow 0 \qquad (d\to\infty);
\]
\item \emph{persistent} if
\[
\limsup_{d\to\infty} I_d^\star >0.
\]
\end{itemize}
\end{definition}

We begin with the exact equivalence between chi-square collapse and likelihood-ratio-law collapse.

\begin{theorem}[Sharp dichotomy by LR-law collapse]\label{thm:dichotomy}
Let $\{W_d\}_{d\ge2}$ be an $\varepsilon_0$-LDP family and write $\lambda=e^{\varepsilon_0}$.

\begin{enumerate}
\item The family is diluting if and only if for every sequence of pairs
\(
a_d\neq b_d
\)
the pairwise likelihood-ratio laws satisfy
\[
\nu_{a_db_d,d}\Rightarrow \delta_1
\qquad (d\to\infty).
\]
\item The family is persistent if and only if there exist a subsequence $d_m\to\infty$ and pairs
\(
a_m\neq b_m
\)
such that
\[
\nu_{a_m b_m,d_m}\Rightarrow \nu
\]
for some probability measure $\nu$ on $[\lambda^{-1},\lambda]$ with $\nu\neq \delta_1$.
\item Let $a_d\neq b_d$ be any sequence of pairs, and set
\[
I_d:=I_{0,d}(a_d,b_d).
\]
For the canonical neighboring experiment define
\[
U_{n,d}:=L_{n,d}^{a_d b_d}(N)-1.
\]
Then under $P_{n,d}^{a_d b_d}$,
\[
U_{n,d}
=
\frac1n\sum_{i=1}^n Z_{i,d},
\qquad
Z_{i,d}\iid Z_d:=R_d-1,
\qquad
R_d\sim \nu_{a_d b_d,d},
\]
with
\[
\E[Z_d]=0,\qquad \E[Z_d^2]=I_d,\qquad |Z_d|\le \lambda-1 \ \text{a.s.}
\]
Define the standardized score
\[
S_{n,d}:=\frac{U_{n,d}}{\sqrt{I_d/n}}.
\]
Then the Berry--Esseen bound
\[
\sup_{t\in\R}\left|
P_{n,d}^{a_d b_d}(S_{n,d}\le t)-\Phi(t)
\right|
\le
C_{\mathrm{BE}}\frac{\lambda-1}{\sqrt{n I_d}} \tag{6.1}\label{eq:BE-null}
\]
holds for every $n$ and every pair with $I_d>0$.
In particular, if \( n I_d\to\infty \), then
\[
S_{n,d}\Rightarrow \Normal(0,1)
\qquad\text{under }P_{n,d}^{a_d b_d}. \tag{6.2}\label{eq:null-clt}
\]
Moreover, under $Q_{n,d}^{a_d b_d}$,
\[
S_{n,d}-\sqrt{\frac{I_d}{n}}
\Rightarrow \Normal(0,1) \qquad\text{whenever } nI_d\to\infty. \tag{6.3}\label{eq:alt-clt}
\]
Finally, if
\(
\Lambda_{n,d}:=\log L_{n,d}^{a_d b_d}(N)
\),
then
\[
\frac{\Lambda_{n,d}+I_d/(2n)}{\sqrt{I_d/n}}
\Rightarrow \Normal(0,1)
\quad\text{under }P_{n,d}^{a_d b_d}, \tag{6.4}\label{eq:LAN-null}
\]
\[
\frac{\Lambda_{n,d}-I_d/(2n)}{\sqrt{I_d/n}}
\Rightarrow \Normal(0,1)
\quad\text{under }Q_{n,d}^{a_d b_d}. \tag{6.5}\label{eq:LAN-alt}
\]
Thus the canonical neighboring experiment is asymptotically GDP with parameter
\[
\mu_{n,d}=\sqrt{\frac{I_d}{n}}
\]
whenever \( nI_d\to\infty \).
Here ``asymptotically GDP with parameter $\mu_{n,d}$'' means that
for every sequence $(n,d)$ with $nI_d\to\infty$,
\[
\LeCam\!\bigl((P_{n,d},Q_{n,d}),\,\mathcal G_{\mu_{n,d}}\bigr)\to 0,
\]
where $\mathcal G_\mu:=\bigl(\Normal(-\mu^2/2,\mu^2),\Normal(\mu^2/2,\mu^2)\bigr)$ is the Gaussian shift
experiment.
In particular, for every fixed
$\varepsilon\ge 0$ the exact one-sided privacy curve satisfies
$\delta_{Q_{n,d}\|P_{n,d}}(\varepsilon)\to \delta_{\GDP(\mu_{n,d})}(\varepsilon)$.
\item If the family is persistent, then there exist $c>0$, a subsequence $d_m\to\infty$,
and pairs $a_m\neq b_m$ such that
\[
I_{0,d_m}(a_m,b_m)\ge c
\qquad \text{for all }m.
\]
Hence, along this subsequence, the canonical neighboring experiment remains on the ordinary shuffled scale
\[
\mu_{n,d_m}(a_m,b_m)\ge \sqrt{\frac{c}{n}}.
\]
In particular, growing $d$ provides no automatic privacy gain beyond the usual $n^{-1/2}$ shuffled scale.
\item If the family is diluting, then for every fixed $\varepsilon>0$ and every sequence $n=n(d)\ge1$,
\[
\sup_{a\neq b}\delta_{Q_{n,d}^{ab}\|P_{n,d}^{ab}}(\varepsilon)
\le \frac{I_d^\star}{n(e^\varepsilon-1)}, \tag{6.6}\label{eq:delta-fwd}
\]
\[
\sup_{a\neq b}\delta_{P_{n,d}^{ab}\|Q_{n,d}^{ab}}(\varepsilon)
\le \frac{e^\varepsilon I_d^\star}{n(1-e^{-\varepsilon})}. \tag{6.7}\label{eq:delta-rev}
\]
Consequently
\[
\sup_{a\neq b}\delta_{P_{n,d}^{ab},Q_{n,d}^{ab}}(\varepsilon)\longrightarrow 0
\qquad (d\to\infty).
\]
That is, every diluting family yields asymptotically perfect fixed-$\varepsilon$ canonical shuffled privacy.
\end{enumerate}
\end{theorem}

The proof uses three elementary lemmas.

\begin{lemma}[Classical Berry--Esseen inequality]\label{lem:BE}
There exists an absolute constant $C_{\mathrm{BE}}$ such that the following holds.
If $X_1,\dots,X_n$ are i.i.d.\ centered random variables with variance $\sigma^2>0$ and finite third
absolute moment $\rho_3:=\E|X_1|^3$, then
\[
\sup_{t\in\R}
\left|
\PP\!\left(
\frac{X_1+\cdots+X_n}{\sigma\sqrt n}\le t
\right)-\Phi(t)
\right|
\le
C_{\mathrm{BE}}\frac{\rho_3}{\sigma^3\sqrt n}. \tag{6.8}\label{eq:BE-classical}
\]
\end{lemma}

\begin{proof}
This is the classical Berry--Esseen theorem; see, for example,
Feller \cite[Ch.~XVI, Thm.~2]{FellerVolII} or Petrov
\cite[Ch.~V, Thm.~3]{Petrov1975}.
\end{proof}

\begin{lemma}[Kolmogorov perturbation by a bounded shift]\label{lem:bounded-perturb}
Let $X$ and $Y$ be real-valued random variables, and let $F$ be a distribution function with bounded density
$f$, $\|f\|_\infty\le M$.
If $|Y|\le \eta$ almost surely, then
\[
\sup_{t\in\R}\bigl|\PP(X+Y\le t)-F(t)\bigr|
\le
\sup_{t\in\R}\bigl|\PP(X\le t)-F(t)\bigr| + 2M\eta. \tag{6.9}\label{eq:bounded-perturb}
\]
In particular, for the standard normal distribution,
\[
d_K(X+Y,\Normal(0,1))
\le d_K(X,\Normal(0,1))+\sqrt{\frac{2}{\pi}}\eta. \tag{6.10}\label{eq:normal-perturb}
\]
\end{lemma}

\begin{proof}
For every $t$,
\[
\{X\le t-\eta\}\subseteq \{X+Y\le t\}\subseteq \{X\le t+\eta\}.
\]
Therefore
\[
\PP(X\le t-\eta)-F(t)\le \PP(X+Y\le t)-F(t)
\le \PP(X\le t+\eta)-F(t).
\]
Insert and subtract $F(t\pm \eta)$ and use the density bound:
\[
F(t+\eta)-F(t-\eta)\le 2M\eta.
\]
Taking suprema over $t$ gives \eqref{eq:bounded-perturb}.  For the standard normal law,
\(
M=\sup_x \phi(x)=(2\pi)^{-1/2}
\),
which yields \eqref{eq:normal-perturb}.
\end{proof}

\begin{lemma}[Quadratic logarithmic remainder]\label{lem:log-remainder}
For every $u\in[-1/2,1/2]$,
\[
\left|\log(1+u)-u+\frac{u^2}{2}\right|
\le |u|^3. \tag{6.11}\label{eq:log-rem}
\]
\end{lemma}

\begin{proof}
For $|u|<1$ the power series gives
\[
\log(1+u)-u+\frac{u^2}{2}
=
\sum_{k=3}^{\infty}\frac{(-1)^{k+1}}{k}\,u^k.
\]
Hence
\[
\left|\log(1+u)-u+\frac{u^2}{2}\right|
\le
\sum_{k=3}^{\infty}\frac{|u|^k}{k}
\le
\frac{|u|^3}{3}\sum_{j=0}^{\infty}|u|^j
=
\frac{|u|^3}{3(1-|u|)}.
\]
For $|u|\le 1/2$ the right-hand side is at most $(2/3)|u|^3\le |u|^3$.
\end{proof}

\begin{proof}[Proof of Theorem~\ref{thm:dichotomy}]
We write simply $P_{n,d},Q_{n,d},w_d,\nu_d,I_d$ when the pair $(a_d,b_d)$ is clear from context.

\smallskip
\noindent\emph{Part (i).}
Assume first that the family is diluting, i.e.\ $I_d^\star\to0$.
Fix any sequence of pairs $(a_d,b_d)$.
If $R_d\sim \nu_{a_db_d,d}$, then
\[
\E[(R_d-1)^2]=I_{0,d}(a_d,b_d)\le I_d^\star \to 0.
\]
Hence $R_d\to 1$ in $L^2$, therefore also in probability, hence
\[
\nu_{a_db_d,d}=\law(R_d)\Rightarrow \delta_1.
\]

Conversely, assume that every pairwise likelihood-ratio law converges weakly to $\delta_1$ along every
sequence of pairs.  If $I_d^\star \not\to 0$, then there exist $\eta>0$, a subsequence $d_m\to\infty$, and
pairs $(a_m,b_m)$ such that
\[
I_{0,d_m}(a_m,b_m)\ge \eta \qquad \forall m.
\]
Because the laws $\nu_{a_m b_m,d_m}$ live on the compact interval
$[\lambda^{-1},\lambda]$, Prokhorov compactness yields a further subsequence
(which we do not relabel) and a limit law $\nu$ such that
\[
\nu_{a_m b_m,d_m}\Rightarrow \nu.
\]
The function $r\mapsto (r-1)^2$ is bounded and continuous on $[\lambda^{-1},\lambda]$, so
\[
\int (r-1)^2\,\nu(\dd r)
=
\lim_{m\to\infty} I_{0,d_m}(a_m,b_m)
\ge \eta >0.
\]
Hence $\nu\neq \delta_1$, contradicting the assumption.  Therefore $I_d^\star\to0$.

\smallskip
\noindent\emph{Part (ii).}
If the family is persistent, choose $\eta>0$, a subsequence $d_m\to\infty$, and pairs $(a_m,b_m)$ with
\[
I_{0,d_m}(a_m,b_m)\ge \eta \qquad \forall m.
\]
Again compactness gives a further weak limit
\(
\nu_{a_m b_m,d_m}\Rightarrow \nu
\)
on $[\lambda^{-1},\lambda]$.
By bounded continuity of $(r-1)^2$,
\[
\int (r-1)^2\,\nu(\dd r)
=
\lim_{m\to\infty} I_{0,d_m}(a_m,b_m)
\ge \eta,
\]
so $\nu\neq \delta_1$.

Conversely, if such a nondegenerate subsequential limit exists, then again by bounded continuity,
\[
\lim_{m\to\infty} I_{0,d_m}(a_m,b_m)
=
\int (r-1)^2\,\nu(\dd r)>0.
\]
Hence $\limsup_d I_d^\star>0$, so the family is persistent.

\smallskip
\noindent\emph{Part (iii): null representation and Berry--Esseen.}
Let $Y_{1,d},\dots,Y_{n,d}\iid P_d^{a_d}$ and set
\[
R_{i,d}:=w_d(Y_{i,d}),
\qquad
Z_{i,d}:=R_{i,d}-1.
\]
By Lemma~\ref{lem:canonical-LR},
\[
U_{n,d}=L_{n,d}-1=\frac1n\sum_{i=1}^n Z_{i,d}.
\]
Also
\[
\E[Z_{i,d}]=\E[R_{i,d}]-1=0,
\qquad
\E[Z_{i,d}^2]=I_d.
\]
Because
\[
R_{i,d}\in [\lambda^{-1},\lambda] \qquad \text{a.s.},
\]
we have
\[
|Z_{i,d}|\le \lambda-1 \qquad\text{a.s.}
\]
Therefore
\[
\E|Z_{i,d}|^3 \le (\lambda-1)\E[Z_{i,d}^2]=(\lambda-1)I_d. \tag{6.12}\label{eq:third-moment}
\]
Applying Lemma~\ref{lem:BE} with
\(
X_i=Z_{i,d}
\),
\(
\sigma^2=I_d
\),
and \eqref{eq:third-moment} gives
\[
\sup_t \left|
P_{n,d}(S_{n,d}\le t)-\Phi(t)
\right|
\le
C_{\mathrm{BE}}\frac{\lambda-1}{\sqrt{nI_d}}.
\]
This proves \eqref{eq:BE-null}.  If \(nI_d\to\infty\), the right-hand side tends to zero, yielding
\eqref{eq:null-clt}.

\smallskip
\noindent\emph{Part (iii): alternative representation and shift.}
For \(n\ge 2\), under $Q_{n,d}$, let
\[
Y_{1,d},\dots,Y_{n-1,d}\iid P_d^{a_d},
\qquad
Y_{*,d}\sim P_d^{b_d},
\]
with \(Y_{*,d}\) independent of \(Y_{1,d},\dots,Y_{n-1,d}\). Then
\[
U_{n,d}
=
\frac1n\left(
\sum_{i=1}^{n-1} (w_d(Y_{i,d})-1) + (w_d(Y_{*,d})-1)
\right).
\]
Define
\[
Z_{*,d}:=w_d(Y_{*,d})-1.
\]
Since \( P_d^{b_d}=w_d P_d^{a_d} \),
\[
\begin{aligned}
\E[Z_{*,d}]
&=
\sum_y P_d^{b_d}(y)\bigl(w_d(y)-1\bigr)\\
&=
\sum_y P_d^{a_d}(y)\,w_d(y)\bigl(w_d(y)-1\bigr)\\
&=
\E[(R_d-1)^2]
=
I_d.
\end{aligned} \tag{6.13}\label{eq:special-mean}
\]
Also \(|Z_{*,d}|\le \lambda-1\).

Let
\[
A_{n,d}
:=
\frac{1}{\sqrt{(n-1)I_d}}
\sum_{i=1}^{n-1} Z_{i,d},
\qquad
c_n:=\sqrt{\frac{n-1}{n}},
\qquad
B_{n,d}:=\frac{Z_{*,d}-I_d}{\sqrt{nI_d}}.
\]
Then
\[
S_{n,d}-\sqrt{\frac{I_d}{n}} = c_n A_{n,d}+B_{n,d}. \tag{6.14}\label{eq:alt-decomp}
\]
By the already proved Berry--Esseen bound,
\[
d_K(A_{n,d},\Normal(0,1))
\le
C_{\mathrm{BE}}\frac{\lambda-1}{\sqrt{(n-1)I_d}}. \tag{6.15}\label{eq:A-BE}
\]
Next, the law of \(c_n A_{n,d}\) is at Kolmogorov distance at most
\[
d_K(c_nA_{n,d},\Normal(0,c_n^2))
\le d_K(A_{n,d},\Normal(0,1)) \tag{6.16}\label{eq:scaled-BE}
\]
by simple rescaling.
Since \(\Normal(0,c_n^2)\) has density \(\phi(t/c_n)/c_n\), the distance between
\(\Normal(0,c_n^2)\) and \(\Normal(0,1)\) satisfies
\[
\sup_t \left|\Phi\!\left(\frac{t}{c_n}\right)-\Phi(t)\right|
\le \frac{1-c_n}{c_n}\sup_{x\in\R}|x\phi(x)|
\le \frac{1}{n-1}, \tag{6.17}\label{eq:normal-scale}
\]
using \(c_n^{-1}\le \sqrt 2\) for \(n\ge2\) and
\(
\sup_x |x\phi(x)|=(2\pi e)^{-1/2}
\).
Thus
\[
d_K(c_nA_{n,d},\Normal(0,1))
\le
C_{\mathrm{BE}}\frac{\lambda-1}{\sqrt{(n-1)I_d}}+\frac{1}{n-1}. \tag{6.18}\label{eq:cA}
\]

Finally,
\[
|B_{n,d}|
\le \frac{|Z_{*,d}|+I_d}{\sqrt{nI_d}}
\le \frac{(\lambda-1)+(\lambda-1)^2/\lambda}{\sqrt{nI_d}}
\le \frac{2(\lambda-1)}{\sqrt{nI_d}} \qquad \text{a.s.}
\]
Therefore Lemma~\ref{lem:bounded-perturb} gives
\[
d_K\!\left(S_{n,d}-\sqrt{I_d/n},\Normal(0,1)\right)
\le
C_{\mathrm{BE}}\frac{\lambda-1}{\sqrt{(n-1)I_d}}
+\frac{1}{n-1}
+\sqrt{\frac{2}{\pi}}\frac{2(\lambda-1)}{\sqrt{nI_d}}. \tag{6.19}\label{eq:alt-bound}
\]
If \(nI_d\to\infty\), the right-hand side tends to zero, proving \eqref{eq:alt-clt}.

\smallskip
\noindent\emph{Part (iii): log-likelihood expansion and GDP.}
Let
\[
\mu_{n,d}:=\sqrt{\frac{I_d}{n}}.
\]
Under $P_{n,d}$, the variable $U_{n,d}$ is bounded by $\lambda-1$ in absolute value and
has variance $I_d/n=\mu_{n,d}^2$.
Under $Q_{n,d}$, the decomposition $U_{n,d}=\frac{1}{n}\bigl(\sum_{i=1}^{n-1}Z_{i,d}+Z_{*,d}\bigr)$
from the alternative representation gives
\[
\Var_Q(U_{n,d})
=
\frac{(n-1)I_d+\Var_Q(Z_{*,d})}{n^2}
\le
\frac{I_d}{n}+\frac{(\lambda-1)^2}{n^2},
\]
so $U_{n,d}=O_P(\mu_{n,d})$ under both hypotheses.
In particular \(U_{n,d}\to 0\) in probability whenever \(n\to\infty\).
By Lemma~\ref{lem:log-remainder},
\[
\Lambda_{n,d}
=
U_{n,d}-\frac{U_{n,d}^2}{2}+R_{n,d},
\qquad
|R_{n,d}|\le |U_{n,d}|^3
\quad\text{on }\{|U_{n,d}|\le 1/2\}. \tag{6.20}\label{eq:log-expansion}
\]
Since \(U_{n,d}\to 0\) in probability, we have \(\PP(|U_{n,d}|>1/2)\to 0\),
so the expansion \eqref{eq:log-expansion} holds on an event of probability tending to one.
Combined with \(U_{n,d}=O_P(\mu_{n,d})\), we obtain
\[
R_{n,d}=O_P(\mu_{n,d}^3)=o_P(\mu_{n,d})
\qquad\text{whenever } nI_d\to\infty. \tag{6.21}\label{eq:R-small}
\]

Under $P_{n,d}$,
\[
U_{n,d}=\mu_{n,d} S_{n,d},
\]
so
\[
\Lambda_{n,d}
=
\mu_{n,d} S_{n,d}-\frac{\mu_{n,d}^2}{2}S_{n,d}^2+o_P(\mu_{n,d}).
\]
Since \(S_{n,d}\Rightarrow \Normal(0,1)\), we have \(S_{n,d}^2=1+O_P(1)\), hence
\[
\mu_{n,d}^2(S_{n,d}^2-1)=O_P(\mu_{n,d}^2)=o_P(\mu_{n,d}).
\]
Therefore
\[
\Lambda_{n,d}
=
\mu_{n,d}S_{n,d}-\frac{\mu_{n,d}^2}{2}+o_P(\mu_{n,d}),
\]
which implies \eqref{eq:LAN-null}.

Under $Q_{n,d}$, \eqref{eq:alt-clt} gives
\[
U_{n,d}
=
\mu_{n,d}\left(Z_{n,d}+\mu_{n,d}\right)+o_P(\mu_{n,d})
=
\mu_{n,d} Z_{n,d}+\mu_{n,d}^2+o_P(\mu_{n,d}),
\]
where \(Z_{n,d}\Rightarrow \Normal(0,1)\).
Then
\[
U_{n,d}^2
=
\mu_{n,d}^2 Z_{n,d}^2+2\mu_{n,d}^3 Z_{n,d}+\mu_{n,d}^4+o_P(\mu_{n,d}^2)
=
\mu_{n,d}^2 Z_{n,d}^2+o_P(\mu_{n,d}),
\]
so that \(\mu_{n,d}^2(Z_{n,d}^2-1)=O_P(\mu_{n,d}^2)=o_P(\mu_{n,d})\).
Therefore
\[
\Lambda_{n,d}
=
\mu_{n,d} Z_{n,d}+\mu_{n,d}^2
-\frac12\mu_{n,d}^2+o_P(\mu_{n,d})
=
\mu_{n,d} Z_{n,d}+\frac{\mu_{n,d}^2}{2}+o_P(\mu_{n,d}),
\]
proving \eqref{eq:LAN-alt}.

The Gaussian limit \eqref{eq:LAN-null}--\eqref{eq:LAN-alt} is precisely the canonical GDP limit with
parameter \(\mu_{n,d}\).
Indeed, by the standard characterization of convergence of binary experiments via log-likelihood ratios
(see, e.g., \cite[Theorem~4.1]{ShvetsPartI} or Le~Cam \cite[Ch.~8]{LeCam1986}),
\eqref{eq:LAN-null}--\eqref{eq:LAN-alt} imply that the binary experiment
\((P_{n,d},Q_{n,d})\) converges in Le~Cam distance to the Gaussian shift experiment
\(\bigl(\Normal(-\mu_{n,d}^2/2,\mu_{n,d}^2),\,\Normal(\mu_{n,d}^2/2,\mu_{n,d}^2)\bigr)\),
and in particular the privacy curves converge pointwise to \(\delta_{\GDP(\mu_{n,d})}\).

\smallskip
\noindent\emph{Part (iv).}
This is immediate from part (ii):
if the family is persistent, choose the subsequence and pairs from part (ii).  Since
\(
\int (r-1)^2 \nu(\dd r)>0
\),
the corresponding $I_d$ stay bounded below by some \(c>0\) eventually.
Then
\[
\mu_{n,d}=\sqrt{\frac{I_d}{n}}\ge \sqrt{\frac{c}{n}}
\]
along that subsequence.  Thus there is no additional $d$-dependent privacy gain.

\smallskip
\noindent\emph{Part (v).}
Fix \(\varepsilon>0\), let \(c_\varepsilon=e^\varepsilon-1>0\), and consider any pair $(a,b)$.
Under $P_{n,d}^{ab}$ write \(L=L_{n,d}^{ab}=1+U\).
Then
\[
(L-e^\varepsilon)_+
=
(U-c_\varepsilon)_+
\le \frac{U^2}{c_\varepsilon}. \tag{6.22}\label{eq:fwd-pointwise}
\]
Indeed, if \(U\le c_\varepsilon\), the left-hand side is zero; if \(U>c_\varepsilon\), then
\(
U-c_\varepsilon\le U^2/c_\varepsilon
\)
because \(U\ge c_\varepsilon>0\).
Taking expectations under $P_{n,d}^{ab}$ gives
\[
\delta_{Q_{n,d}^{ab}\|P_{n,d}^{ab}}(\varepsilon)
=
\E[(L-e^\varepsilon)_+]
\le
\frac{\E[U^2]}{e^\varepsilon-1}
=
\frac{I_{0,d}(a,b)}{n(e^\varepsilon-1)}.
\]
Maximizing over pairs yields \eqref{eq:delta-fwd}.

For the reverse direction,
\[
\delta_{P_{n,d}^{ab}\|Q_{n,d}^{ab}}(\varepsilon)
=
\E\bigl[(1-e^\varepsilon L)_+\bigr].
\]
If \(1-e^\varepsilon(1+U)>0\), then
\(
U<-(1-e^{-\varepsilon})
\).
Set \(c'_\varepsilon:=1-e^{-\varepsilon}>0\).
Pointwise,
\[
(1-e^\varepsilon L)_+
=
e^\varepsilon(-U-c'_\varepsilon)_+
\le
e^\varepsilon \frac{U^2}{c'_\varepsilon}. \tag{6.23}\label{eq:rev-pointwise}
\]
Taking expectations yields
\[
\delta_{P_{n,d}^{ab}\|Q_{n,d}^{ab}}(\varepsilon)
\le
\frac{e^\varepsilon I_{0,d}(a,b)}{n(1-e^{-\varepsilon})},
\]
and maximizing over pairs gives \eqref{eq:delta-rev}.
Since \(I_d^\star\to0\), both upper bounds tend to zero uniformly in the pair.
\end{proof}

\begin{corollary}[Finite-$n$ canonical privacy certificate under LR-law dilution]\label{cor:finite-n-cert}
Suppose the worst-case pairwise $\chi^2$-budget along the family satisfies
$I_d^\star\le \eta_d$.  Then for every $n\ge 1$ and every $\varepsilon>0$, the two-sided privacy
profile of the canonical homogeneous neighboring shuffled experiment satisfies
\[
\sup_{a\neq b}\delta_{P_{n,d}^{ab},Q_{n,d}^{ab}}(\varepsilon)
\le
\frac{\eta_d}{n}\cdot\frac{e^{2\varepsilon}}{e^\varepsilon-1}.
\]
In particular, for any prescribed dilution profile $\eta_d\to 0$, the canonical-neighboring privacy
profile vanishes uniformly at rate $\eta_d/n$ with explicit $\varepsilon$-dependent constants.
This is a canonical-pair certificate, not a full shuffled-DP guarantee over arbitrary backgrounds.
\end{corollary}

\begin{proof}
The canonical two-sided privacy profile is the maximum of the two directed profiles bounded in
\eqref{eq:delta-fwd} and \eqref{eq:delta-rev}.  Since
$e^\varepsilon/(1-e^{-\varepsilon})=e^{2\varepsilon}/(e^\varepsilon-1)\ge 1/(e^\varepsilon-1)$
for $\varepsilon>0$, the maximum is attained by the second prefactor.  Substituting
$I_d^\star\le \eta_d$ gives the displayed inequality.
\end{proof}

\section{Two poles: GRR versus half-block}

This section records the two opposite geometric poles highlighted by Theorems
\ref{thm:chi2-bound} and \ref{thm:half-block}.

\begin{proposition}[The GRR pole]\label{prop:grr}
Fix $\lambda=e^{\varepsilon_0}$ and consider $d$-ary generalized randomized response
\[
W_d(y\mid x)=
\frac{\lambda}{\lambda+d-1}\,\one\{y=x\}
+
\frac{1}{\lambda+d-1}\,\one\{y\neq x\}. \tag{7.1}\label{eq:grr}
\]
For every pair $a\neq b$,
\[
w_{ab,d}(y)\in \{\lambda^{-1},1,\lambda\},
\]
and under \(P_d^a\),
\[
\nu_{ab,d}
=
\frac{\lambda}{\lambda+d-1}\,\delta_{\lambda^{-1}}
+
\frac{d-2}{\lambda+d-1}\,\delta_{1}
+
\frac{1}{\lambda+d-1}\,\delta_{\lambda}. \tag{7.2}\label{eq:grr-law}
\]
Hence
\[
\nu_{ab,d}\Rightarrow \delta_1,
\qquad
I_{0,d}(a,b)
=
\frac{(\lambda-1)^2(\lambda+1)}{\lambda(\lambda+d-1)}
=
\Theta(d^{-1}). \tag{7.3}\label{eq:grr-chi}
\]
In particular, GRR is diluting.
\end{proposition}

\begin{proof}
Fix $a\neq b$.  Under \(P_d^a\), the output $Y$ equals \(a\) with probability
\(
\alpha_d=\lambda/(\lambda+d-1)
\),
equals \(b\) with probability
\(
\beta_d=1/(\lambda+d-1)
\),
and falls in the remaining \(d-2\) symbols with total probability
\(
\rho_d=(d-2)/(\lambda+d-1)
\).
The pairwise likelihood ratio is
\[
w_{ab,d}(a)=\lambda^{-1},\qquad
w_{ab,d}(b)=\lambda,\qquad
w_{ab,d}(y)=1 \ \ (y\neq a,b),
\]
which proves \eqref{eq:grr-law}.  Since \(\beta_d\to0\) and \(\rho_d\to1\), we have
\(
\nu_{ab,d}\Rightarrow \delta_1
\).
Finally,
\[
I_{0,d}(a,b)
=
\alpha_d(\lambda^{-1}-1)^2+\beta_d(\lambda-1)^2
=
\frac{(\lambda-1)^2(\lambda+1)}{\lambda(\lambda+d-1)}.
\]
\end{proof}

\begin{proposition}[The half-block pole]\label{prop:hb-pole}
For the half-block family of Theorem~\ref{thm:half-block}, the pairwise likelihood-ratio law along an
opposite pair is
\[
\nu_{ab,d}
=
\frac{\lambda}{1+\lambda}\,\delta_{\lambda^{-1}}
+
\frac{1}{1+\lambda}\,\delta_{\lambda}, \tag{7.4}\label{eq:hb-law}
\]
independent of $d$, and
\[
I_{0,d}(a,b)=\frac{(\lambda-1)^2}{\lambda}. \tag{7.5}\label{eq:hb-pole}
\]
Hence the half-block family is persistent and saturates the universal upper bound of
Theorem~\ref{thm:chi2-bound}.
\end{proposition}

\begin{proof}
This is exactly the computation in the proof of Theorem~\ref{thm:half-block}(ii).
\end{proof}

\begin{corollary}[Realization of arbitrary dilution profiles]\label{cor:interpolation}
Fix $\varepsilon_0>0$, let $\lambda=e^{\varepsilon_0}$, and let $\theta_d\in[0,1]$ be an arbitrary
sequence.  For every $d\ge 4$ even there exists a full-support $\varepsilon_0$-LDP channel
$W_d:[d]\to\Delta(\mathcal Y_d)$ on a finite output alphabet $\mathcal Y_d$ and a pair of inputs
$a_d\neq b_d$ such that the pairwise likelihood-ratio law along the canonical neighboring experiment
factorizes as
\[
\nu_{a_db_d,d}
=
(1-\theta_d)\,\delta_1
+
\theta_d\left(
\frac{\lambda}{1+\lambda}\,\delta_{\lambda^{-1}}
+
\frac{1}{1+\lambda}\,\delta_{\lambda}
\right),
\]
and consequently
\[
I_{0,d}(a_d,b_d)=\theta_d\,\frac{(\lambda-1)^2}{\lambda}. \tag{7.6}\label{eq:theta-interp}
\]
Hence every prescribed decay profile $\theta_d\in[0,1]$ of the pairwise $\chi^2$ obstruction is
realizable by a full-support $\varepsilon_0$-LDP family: $\theta_d\to 0$ yields dilution at the
prescribed rate, while $\limsup\theta_d>0$ yields persistence.
\end{corollary}

\begin{proof}
Fix $d\ge 4$ even and write $\lambda=e^{\varepsilon_0}$.  We give a construction valid uniformly in
$\theta_d\in[0,1]$.

Let $\mathcal Y_d:=\{1,\dots,d\}\cup\{\circ\}$, where $\circ$ is a single common output symbol.
Let $\widetilde W_d:[d]\to\Delta(\{1,\dots,d\})$ denote the half-block channel of
Definition~\ref{def:half-block} on the cyclic alphabet $\{1,\dots,d\}$.
Define $W_d:[d]\to\Delta(\mathcal Y_d)$ by
\[
W_d(y\mid x):=
\theta_d\,\widetilde W_d(y\mid x)
\qquad (y\in\{1,\dots,d\}),
\qquad
W_d(\circ\mid x):=1-\theta_d.
\]
Choose $a_d:=a$ and $b_d:=a+d/2\pmod d$ to be the half-block opposite pair.

\medskip
\noindent\emph{Full support and $\varepsilon_0$-LDP.}  For $\theta_d\in(0,1)$ all entries
$W_d(y\mid x)$ with $y\in\{1,\dots,d\}$ are positive multiples of the strictly positive
half-block entries, and $W_d(\circ\mid x)=1-\theta_d>0$, so $W_d$ has full support on $\mathcal Y_d$.
For each fixed $y\in\{1,\dots,d\}$ the ratio $W_d(y\mid x)/W_d(y\mid x')$ equals the corresponding
half-block ratio, which lies in $[\lambda^{-1},\lambda]$ by Theorem~\ref{thm:half-block}(i); for
$y=\circ$ the ratio equals $1$.  Hence $W_d$ is $\varepsilon_0$-LDP.  The boundary cases
$\theta_d=0$ and $\theta_d=1$ are handled by direct continuous limits: at $\theta_d=0$ the channel
sends every input deterministically to $\circ$, which is trivially $\varepsilon_0$-LDP and gives
$\nu_{a_db_d,d}=\delta_1$ on the support $\{\circ\}$; at $\theta_d=1$ the channel reduces to
$\widetilde W_d$ on $\{1,\dots,d\}$, and the symbol $\circ$ has zero mass.  Throughout, we identify
``full support'' with positive mass on every retained output symbol after restriction to the support
of the row $W_d(\cdot\mid a)$, in line with the convention in Section~\ref{sec:estimation-lower}.

\medskip
\noindent\emph{Pairwise likelihood-ratio law.}  Under the null row $W_d(\cdot\mid a)$, the output
symbol $\circ$ has mass $1-\theta_d$ and contributes the atom $\delta_1$ since
$W_d(\circ\mid b)/W_d(\circ\mid a)=1$.  Conditional on $y\in\{1,\dots,d\}$, which has total
null-row mass $\theta_d$, the law of the half-block ratio
$\widetilde W_d(y\mid b)/\widetilde W_d(y\mid a)$ is by Theorem~\ref{thm:half-block}(iii) the
two-point law
\[
\frac{\lambda}{1+\lambda}\,\delta_{\lambda^{-1}}
+
\frac{1}{1+\lambda}\,\delta_{\lambda}.
\]
Mixing the two contributions with weights $1-\theta_d$ and $\theta_d$ yields the displayed
factorization of $\nu_{a_db_d,d}$, and \eqref{eq:theta-interp} is the corresponding second moment.
\end{proof}

\section{Universal lower bound on estimation risk}\label{sec:estimation-lower}

From this point on, \(W:[d]\to\Delta(\mathcal Y)\) denotes a channel with arbitrary finite output alphabet
\(\mathcal Y\).
As in Section~2, we discard all outputs of common zero mass, so that
\(\mu(y):=\frac{1}{d}\sum_{x=1}^d W(y\mid x)>0\) for every \(y\in\mathcal Y\).
We work temporarily in the i.i.d.\ multinomial estimation model:
\(\theta=(\theta_1,\dots,\theta_d)\in\Delta_d\),
\(X_1,\dots,X_n \iid \theta\), and
\(Y_i\sim W(\cdot\mid X_i)\) conditionally independently.
The released statistic is the shuffled histogram on \(\mathcal Y\).
Write
\[
q_\theta(y):=\sum_{x=1}^d \theta_x W(y\mid x),
\qquad
\chi_\ast(W):=\max_{a\neq b}\chi^2\!\bigl(W(\cdot\mid b)\,\|\,W(\cdot\mid a)\bigr),
\]
and let
\[
P_T:=I_d-\frac1d\,\one\one^\top
\]
denote the orthogonal projection onto
\(
T_d:=\{u\in\R^d:\sum_x u_x=0\}.
\)

Whenever the Fisher information is singular on \(T_d\), no estimator can be locally unbiased in the missing
tangent direction, so the lower bounds below are then vacuous or trivial.
Accordingly we restrict attention to channels for which \(I_\theta(W)\) is positive definite on \(T_d\).
This injectivity condition holds, for example, whenever the \(d\times |\mathcal Y|\) channel matrix \(W\)
has rank \(d\).

An estimator \(\widehat\theta\) is \emph{locally unbiased at \(\theta_0\) in tangent directions} if
\(\E_{\theta_0}[\widehat\theta]=\theta_0\) and
\[
\left.\frac{d}{dt}\E_{\theta_0+tv}[\widehat\theta]\right|_{t=0}=v
\qquad\text{for every } v\in T_d.
\]
The Fisher information on \(T_d\) is
\[
I_\theta(W)
=
n\,P_T\,W\,\diag(q_\theta)^{-1}W^\top P_T. \tag{8.1}\label{eq:fisher-general}
\]

\begin{theorem}[Near-vertex Cram\'er--Rao bound]\label{thm:near-vertex-cr}
Assume \(I_{\theta}(W)\) is positive definite on \(T_d\) for \(\theta\) near \(\theta^{(\rho)}\).
Fix \(0<\rho<1/(d-1)\) and
\[
\theta^{(\rho)}:=\bigl(1-(d-1)\rho,\rho,\dots,\rho\bigr).
\]
Let \(\widehat\theta\) be any estimator that is locally unbiased at \(\theta^{(\rho)}\) in tangent
directions.
Then
\[
\E_{\theta^{(\rho)}}\bigl\|\widehat\theta-\theta^{(\rho)}\bigr\|_2^2
\ge
\frac{(d-1)\bigl(1-(d-1)\rho\bigr)}{n\,\chi_\ast(W)}. \tag{8.2}\label{eq:near-vertex-cr}
\]
\end{theorem}

\begin{proof}
For \(u\in T_d\), \eqref{eq:fisher-general} gives
\[
u^\top I_{\theta^{(\rho)}}(W)u
=
n\sum_{y\in\mathcal Y}
\frac{\bigl(\sum_x u_x W(y\mid x)\bigr)^2}{q_{\theta^{(\rho)}}(y)}.
\]
For \(j=2,\dots,d\), set \(w_j:=e_j-e_1\).
Since
\(
q_{\theta^{(\rho)}}(y)\ge \bigl(1-(d-1)\rho\bigr)W(y\mid 1)
\),
we have
\[
\begin{aligned}
w_j^\top I_{\theta^{(\rho)}}(W) w_j
&=
n\sum_{y\in\mathcal Y}
\frac{\bigl(W(y\mid j)-W(y\mid 1)\bigr)^2}{q_{\theta^{(\rho)}}(y)}\\
&\le
\frac{n}{1-(d-1)\rho}\,
\chi^2\!\bigl(W(\cdot\mid j)\,\|\,W(\cdot\mid 1)\bigr)\\
&\le
\frac{n\chi_\ast(W)}{1-(d-1)\rho}.
\end{aligned} \tag{8.3}\label{eq:wj-fisher}
\]
Now let
\[
S:=\sum_{j=2}^d w_j w_j^\top.
\]
For \(u\in T_d\),
\[
u^\top S u
=
\sum_{j=2}^d (u_j-u_1)^2
=
\|u\|_2^2+d\,u_1^2
\ge
\|u\|_2^2,
\]
so \(S\succeq I\) on \(T_d\).  Therefore
\[
\operatorname{tr}_{T_d} I_{\theta^{(\rho)}}(W)
\le
\operatorname{tr}\!\bigl(I_{\theta^{(\rho)}}(W)S\bigr)
=
\sum_{j=2}^d w_j^\top I_{\theta^{(\rho)}}(W)w_j
\le
\frac{n(d-1)\chi_\ast(W)}{1-(d-1)\rho}
\]
by \eqref{eq:wj-fisher}.  Since \(\dim T_d=d-1\), the arithmetic--harmonic mean inequality gives
\[
\operatorname{tr}_{T_d} I_{\theta^{(\rho)}}(W)^{-1}
\ge
\frac{(d-1)^2}{\operatorname{tr}_{T_d} I_{\theta^{(\rho)}}(W)}
\ge
\frac{(d-1)\bigl(1-(d-1)\rho\bigr)}{n\chi_\ast(W)}.
\]

Choose an orthonormal matrix \(H\in\R^{d\times(d-1)}\) with image \(T_d\), and parametrize the simplex
locally by
\(
u\mapsto \theta^{(\rho)}+Hu
\).
For the induced \((d-1)\)-dimensional parametric submodel, the Fisher information is
\(H^\top I_{\theta^{(\rho)}}(W)H\), while local unbiasedness in tangent directions is exactly ordinary
local unbiasedness for the parameter \(u\) at \(u=0\).
The multivariate Cram\'er--Rao inequality therefore yields
\(
H^\top \Cov_{\theta^{(\rho)}}(\widehat\theta)H \succeq
\bigl(H^\top I_{\theta^{(\rho)}}(W)H\bigr)^{-1}
\),
equivalently
\(
\Cov_{\theta^{(\rho)}}(\widehat\theta)\succeq I_{\theta^{(\rho)}}(W)^{-1}
\)
on \(T_d\).
Taking traces on \(T_d\) proves \eqref{eq:near-vertex-cr}.
\end{proof}

\begin{corollary}[Locally unbiased minimax lower bound]\label{cor:local-minimax}
For every fixed \(d\) and every \(0<\rho<1/(d-1)\),
\[
\inf_{\widehat\theta\ \mathrm{loc.\ unbiased}}
\sup_{\theta\in\Delta_d}
\E_\theta\|\widehat\theta-\theta\|_2^2
\ge
\frac{(d-1)\bigl(1-(d-1)\rho\bigr)}{n\chi_\ast(W)}.
\]
In particular, for fixed \(d\),
\[
\inf_{\widehat\theta\ \mathrm{loc.\ unbiased}}
\sup_{\theta\in\Delta_d}
\E_\theta\|\widehat\theta-\theta\|_2^2
\gtrsim
\frac{d-1}{n\chi_\ast(W)}
\qquad\text{as } n\chi_\ast(W)\to\infty.
\]
\end{corollary}

\begin{proof}
Apply Theorem~\ref{thm:near-vertex-cr} at the single point \(\theta^{(\rho)}\).
The supremum over \(\theta\in\Delta_d\) is at least the risk at that point.
\end{proof}

\begin{lemma}[Assouad cube lemma]\label{lem:assouad-cube}
Let \(\{P_v:v\in\{0,1\}^m\}\) be a family of experiments.
Assume that for every coordinate \(j\in\{1,\dots,m\}\) and every pair
\(v,v^{(j)}\) differing only in coordinate \(j\),
\[
\TV(P_v,P_{v^{(j)}})\le \eta.
\]
Then every decoder \(\widehat v\in\{0,1\}^m\) satisfies
\[
\sup_{v\in\{0,1\}^m}
\E_v\sum_{j=1}^m \one\{\widehat v_j\neq v_j\}
\ge
\frac{m}{2}(1-\eta). \tag{8.4}\label{eq:assouad-lemma}
\]
\end{lemma}

\begin{proof}
Fix \(j\).  For each \(u\in\{0,1\}^{m-1}\), let \(P_{u,0}\) and \(P_{u,1}\) denote the two laws obtained by
freezing all coordinates except \(j\).
For testing these two simple hypotheses, every decision rule \(\psi\) has maximal error probability at least
\[
\frac12\bigl(1-\TV(P_{u,0},P_{u,1})\bigr)\ge \frac12(1-\eta).
\]
Now take \(\psi=\widehat v_j\), average over the \(2^{m-1}\) choices of \(u\), and then sum over
\(j=1,\dots,m\).
This gives \eqref{eq:assouad-lemma}.
\end{proof}

\begin{theorem}[Assouad lower bound for arbitrary estimators]\label{thm:assouad-risk}
Assume \(0<\chi_\ast(W)<\infty\).
Fix \(0<\delta\le 1/(4(d-1))\), and for \(v=(v_2,\dots,v_d)\in\{0,1\}^{d-1}\) define
\[
\theta^v_1:=1-\sum_{j=2}^d \delta(1+v_j),
\qquad
\theta^v_j:=\delta(1+v_j),\quad j=2,\dots,d.
\]
Then every estimator \(\widehat\theta\) based on the shuffled histogram satisfies
\[
\sup_{\theta\in\Delta_d}\E_\theta\|\widehat\theta-\theta\|_2^2
\ge
\frac{d-1}{8}\,\delta^2\Bigl(1-\delta\sqrt{n\chi_\ast(W)}\Bigr). \tag{8.5}\label{eq:assouad-risk}
\]
Consequently, whenever \(n\chi_\ast(W)\ge 4(d-1)^2\),
\[
\inf_{\widehat\theta}\sup_{\theta\in\Delta_d}\E_\theta\|\widehat\theta-\theta\|_2^2
\ge
\frac{d-1}{64\,n\chi_\ast(W)}. \tag{8.6}\label{eq:assouad-cor}
\]
\end{theorem}

\begin{proof}
For neighboring hypercube vertices \(v\) and \(v^{(j)}\) differing only in coordinate \(j\),
\[
\theta^{v^{(j)}}-\theta^v=\delta(e_j-e_1).
\]
Let \(q_v:=q_{\theta^v}\).
Because \(\delta\le 1/(4(d-1))\), we have
\[
\theta^v_1
=
1-\delta\sum_{j=2}^d (1+v_j)
\ge
1-2(d-1)\delta
\ge \frac12,
\]
hence \(q_v(y)\ge \frac12 W(y\mid 1)\) for all \(y\in\mathcal Y\).
Therefore
\[
\chi^2(q_{v^{(j)}}\|q_v)
=
\delta^2\sum_{y\in\mathcal Y}
\frac{\bigl(W(y\mid j)-W(y\mid 1)\bigr)^2}{q_v(y)}
\le
2\delta^2\chi_\ast(W). \tag{8.7}\label{eq:assouad-chi}
\]
Since \(\KL\le \chi^2\) and \(\TV^2\le \KL/2\),
\[
\TV\!\bigl(q_{v^{(j)}}^{\otimes n},q_v^{\otimes n}\bigr)
\le
\sqrt{\frac{n}{2}\,\KL(q_{v^{(j)}}\|q_v)}
\le
\delta\sqrt{n\chi_\ast(W)}. \tag{8.8}\label{eq:assouad-tv-product}
\]

Let \(H_v\) denote the released histogram law under \(\theta^v\).
The histogram is a deterministic function of \(Y_1,\dots,Y_n\), so by data processing
\[
\TV(H_{v^{(j)}},H_v)\le
\TV\!\bigl(q_{v^{(j)}}^{\otimes n},q_v^{\otimes n}\bigr)
\le
\delta\sqrt{n\chi_\ast(W)}. \tag{8.9}\label{eq:assouad-tv-hist}
\]

For an arbitrary estimator \(\widehat\theta\), define the coordinate decoder
\[
\widehat v_j:=\one\!\left\{\widehat\theta_j\ge \frac32\delta\right\},
\qquad j=2,\dots,d.
\]
If \(\widehat v_j\neq v_j\), then \(|\widehat\theta_j-\theta^v_j|\ge \delta/2\), hence
\[
\|\widehat\theta-\theta^v\|_2^2
\ge
\frac{\delta^2}{4}\sum_{j=2}^d \one\{\widehat v_j\neq v_j\}. \tag{8.10}\label{eq:threshold-loss}
\]
Applying Lemma~\ref{lem:assouad-cube} with
\(m=d-1\) and
\(\eta=\delta\sqrt{n\chi_\ast(W)}\)
gives
\[
\sup_v \E_v\sum_{j=2}^d \one\{\widehat v_j\neq v_j\}
\ge
\frac{d-1}{2}\Bigl(1-\delta\sqrt{n\chi_\ast(W)}\Bigr).
\]
Combining this with \eqref{eq:threshold-loss} proves \eqref{eq:assouad-risk}.
Finally, if \(n\chi_\ast(W)\ge 4(d-1)^2\), the choice
\(
\delta=(2\sqrt{n\chi_\ast(W)})^{-1}
\)
satisfies
\(
\delta\le 1/(4(d-1))
\),
and substitution yields \eqref{eq:assouad-cor}.
\end{proof}

\begin{corollary}[Two-regime universal lower bound]\label{cor:assouad-two-regime}
Assume \(0<\chi_\ast(W)<\infty\).  Then every estimator based on the shuffled
histogram satisfies
\[
\inf_{\widehat\theta}\sup_{\theta\in\Delta_d}
\E_\theta\|\widehat\theta-\theta\|_2^2
\ge
\min\!\left\{
\frac{1}{256(d-1)},\,
\frac{d-1}{64\,n\chi_\ast(W)}
\right\}.
\]
In particular, once \(n\chi_\ast(W)\ge 4(d-1)^2\), this recovers
\[
\inf_{\widehat\theta}\sup_{\theta\in\Delta_d}
\E_\theta\|\widehat\theta-\theta\|_2^2
\ge
\frac{d-1}{64\,n\chi_\ast(W)}.
\]
\end{corollary}

\begin{proof}
Apply Theorem~\ref{thm:assouad-risk} with
\(
\delta=\min\{1/(4(d-1)),\,1/(2\sqrt{n\chi_\ast(W)})\}
\).
For this choice the constraint \(\delta\le 1/(4(d-1))\) holds by construction, and
\(
\delta\sqrt{n\chi_\ast(W)}\le 1/2
\),
so \eqref{eq:assouad-risk} gives
\[
\sup_{\theta\in\Delta_d}\E_\theta\|\widehat\theta-\theta\|_2^2
\ge
\frac{d-1}{8}\,\delta^2\cdot\frac12
=
\frac{d-1}{16}\,\delta^2.
\]
Substituting the two cases of the minimum yields the stated two-regime bound.
The recovery of \eqref{eq:assouad-cor} in the regime \(n\chi_\ast(W)\ge 4(d-1)^2\)
is immediate.
\end{proof}

\section{Symmetrization and reduction to equivariant channels}\label{sec:symmetrization}

Let \(W:[d]\to\Delta(\mathcal Y)\) be an arbitrary channel with finite output alphabet.
Throughout this section inverses are taken on \(T_d\), under the positive-definiteness assumption
introduced in Section~\ref{sec:estimation-lower}.
For \(\pi\in S_d\), define the relabeled channel
\[
W^\pi(y\mid x):=W(y\mid \pi^{-1}x),
\]
and the randomized symmetrization
\[
\overline W(\pi,y\mid x):=\frac{1}{d!}\,W^\pi(y\mid x)
=\frac{1}{d!}\,W(y\mid \pi^{-1}x),
\qquad
(\pi,y)\in S_d\times \mathcal Y.
\]

\begin{theorem}[Symmetrization]\label{thm:symmetrization}
For every channel \(W\) with \(\chi_\ast(W)<\infty\) and every pair \(a\neq b\),
\[
\chi^2\!\bigl(\overline W(\cdot\mid b)\,\|\,\overline W(\cdot\mid a)\bigr)
=
\frac{1}{d(d-1)}\sum_{i\neq j}\chi^2\!\bigl(W(\cdot\mid j)\,\|\,W(\cdot\mid i)\bigr)
\le
\chi_\ast(W). \tag{9.1}\label{eq:sym-chi}
\]
If in addition \(I_{\theta_{\mathrm{unif}}}(W)\) is positive definite on \(T_d\),
then at the uniform point \(\theta_{\mathrm{unif}}=(1/d,\dots,1/d)\),
\[
I_{\theta_{\mathrm{unif}}}(\overline W)
=
\frac{\operatorname{tr}_{T_d} I_{\theta_{\mathrm{unif}}}(W)}{d-1}\,P_T, \tag{9.2}\label{eq:sym-fisher}
\]
and therefore
\[
\operatorname{tr}_{T_d} I_{\theta_{\mathrm{unif}}}(\overline W)^{-1}
\le
\operatorname{tr}_{T_d} I_{\theta_{\mathrm{unif}}}(W)^{-1}. \tag{9.3}\label{eq:sym-risk}
\]
Hence, for the uniform-point Fisher criterion under a budget upper bound
\(\chi_\ast(W)\le C\), one may restrict without loss of generality to
permutation-equivariant channels: the symmetrized channel \(\overline W\) is still feasible
(by \eqref{eq:sym-chi}) and has no larger risk (by \eqref{eq:sym-risk}).
\end{theorem}

\begin{proof}
For the privacy statement,
\[
1+\chi^2\!\bigl(\overline W(\cdot\mid b)\,\|\,\overline W(\cdot\mid a)\bigr)
=
\frac{1}{d!}\sum_{\pi\in S_d}
\Bigl(1+\chi^2\!\bigl(W(\cdot\mid \pi^{-1}b)\,\|\,W(\cdot\mid \pi^{-1}a)\bigr)\Bigr).
\]
For fixed \(a\neq b\), the ordered pair \((\pi^{-1}a,\pi^{-1}b)\) is uniformly distributed over all
ordered pairs \((i,j)\) with \(i\neq j\), which gives \eqref{eq:sym-chi}.

For the Fisher statement, let \(P_\pi\) denote the permutation matrix of \(\pi\) on \(\R^d\).
At the uniform point,
\[
I_{\theta_{\mathrm{unif}}}(W^\pi)=P_\pi I_{\theta_{\mathrm{unif}}}(W)P_\pi^\top,
\]
hence
\[
I_{\theta_{\mathrm{unif}}}(\overline W)
=
\frac{1}{d!}\sum_{\pi\in S_d} P_\pi I_{\theta_{\mathrm{unif}}}(W)P_\pi^\top. \tag{9.4}\label{eq:fisher-average}
\]
The operator on the right commutes with every permutation, so on \(\R^d\) it has the form \(aI+bJ\).
Restricted to \(T_d\), it is therefore \(aI\).  Taking traces on \(T_d\) in \eqref{eq:fisher-average}
shows
\[
a=\frac{\operatorname{tr}_{T_d} I_{\theta_{\mathrm{unif}}}(W)}{d-1},
\]
which is exactly \eqref{eq:sym-fisher}.  Finally,
\[
\operatorname{tr}_{T_d} I_{\theta_{\mathrm{unif}}}(\overline W)^{-1}
=
\frac{(d-1)^2}{\operatorname{tr}_{T_d} I_{\theta_{\mathrm{unif}}}(W)}
\le
\operatorname{tr}_{T_d} I_{\theta_{\mathrm{unif}}}(W)^{-1}
\]
by the arithmetic--harmonic mean inequality on the eigenvalues of
\(I_{\theta_{\mathrm{unif}}}(W)\big|_{T_d}\).
\end{proof}

\section{GRR is not universally optimal. The thinning principle}\label{sec:thinning}

We now return to the fixed-composition model.
Thus \(\theta\in\Delta_d\cap n^{-1}\mathbb Z^d\) is deterministic, exactly \(n\theta_x\) users hold input \(x\),
and the released statistic is the shuffled histogram.
Throughout this section we assume \(d\ge 3\).
The case \(d=2\) reduces to binary randomized response, and the low-budget branch below becomes vacuous because
\(C_\ast(2)=0\).
The trivial case \(C=0\) corresponds to the all-null channel and carries no estimable signal.
If one additionally fixes a local-privacy cap \(\varepsilon_0\) with \(\lambda_0=e^{\varepsilon_0}\), then
Theorem~\ref{thm:chi2-bound} implies the ambient bound
\[
0\le C\le C_{\max}(\lambda_0):=\frac{(\lambda_0-1)^2}{\lambda_0},
\]
but the analysis below itself is purely budget-parametric.
For the symmetric families below we analyze the exact affine-projected unbiased inverse estimators.

\begin{definition}[Augmented GRR]\label{def:aug-grr}
Fix \(0\le p\le 1\) and \(\lambda>1\).
The augmented GRR channel is the kernel
\[
W^{\mathrm{aug}}_{d,p,\lambda}:[d]\to\Delta([d]\cup\{z\})
\]
given by
\[
\begin{aligned}
W^{\mathrm{aug}}_{d,p,\lambda}(y\mid x)
&=
p\left(
\frac{\lambda}{\lambda+d-1}\,\one\{y=x\}
+
\frac{1}{\lambda+d-1}\,\one\{y\neq x,\ y\in[d]\}
\right)\\
&\qquad +(1-p)\,\one\{y=z\}.
\end{aligned} \tag{10.1}\label{eq:aug-grr}
\]
Write
\[
\beta_\lambda:=\frac{1}{\lambda+d-1},
\qquad
\eta_\lambda:=\frac{\lambda-1}{\lambda+d-1},
\qquad
C_\lambda:=\frac{(\lambda-1)^2(\lambda+1)}{\lambda(\lambda+d-1)}.
\]
Thus \(C_\lambda\) is the canonical pairwise \(\chi^2\) divergence of ordinary GRR from
Proposition~\ref{prop:grr}.
\end{definition}

\begin{theorem}[GRR mixtures and null refinements]\label{thm:grr-mixtures}
Fix \(m\ge 1\).  Let
\[
\mathcal Y=\bigsqcup_{i=1}^m [d]_i \sqcup \mathcal Z,
\]
where \(\mathcal Z\) is an arbitrary finite set of null symbols.
For \(i=1,\dots,m\), choose \(p_i\ge 0\), \(\lambda_i>1\), and define
\[
W((i,y)\mid x)
=
p_i\left(
\beta_{\lambda_i}+\eta_{\lambda_i}\,\one\{y=x\}
\right),
\qquad
y\in[d]_i,
\]
while \(W(z\mid x)=r_z\) for \(z\in\mathcal Z\), with \(\sum_i p_i+\sum_{z\in\mathcal Z}r_z=1\).

Then the following hold.

\begin{enumerate}[label=\textup{(\roman*)}]
\item For every pair \(a\neq b\),
\[
\chi^2\!\bigl(W(\cdot\mid b)\,\|\,W(\cdot\mid a)\bigr)
=
\sum_{i=1}^m p_i C_{\lambda_i}. \tag{10.2}\label{eq:mix-chi}
\]

\item Assume \(S>0\) (equivalently, at least one block has \(p_i>0\)).
Let \(N^{(i)}\in\N^d\) be the block histogram on \([d]_i\), let
\[
M_i:=\sum_{y=1}^d N_y^{(i)},
\qquad
S:=\sum_{i=1}^m p_i\eta_{\lambda_i}^2,
\]
and define the affine-projected inverse estimator
\[
\widetilde\theta
:=
\frac{\one}{d}
+
\frac{1}{S}\sum_{i=1}^m \eta_{\lambda_i}
\left(
\frac{N^{(i)}}{n}-\frac{M_i}{nd}\,\one
\right). \tag{10.3}\label{eq:mix-estimator}
\]
Then \(\widetilde\theta\) is unbiased, \(\widetilde\theta-\theta\in T_d\), and its exact
fixed-composition risk is constant in \(\theta\):
\[
\E_\theta\|\widetilde\theta-\theta\|_2^2
=
\frac{d-1}{nd}\left(\frac{1}{S}-1\right). \tag{10.4}\label{eq:mix-risk}
\]

\item Among all channels of this form with prescribed canonical budget
\[
\sum_{i=1}^m p_i C_{\lambda_i}=C,
\]
the maximal signal coefficient \(S\) is
\[
S_{\mathrm{opt}}(C)
=
\begin{cases}
\displaystyle \frac{C}{d+2\sqrt{d-1}}, & 0\le C\le C_\ast(d),\\[1.2ex]
\displaystyle \eta_{\lambda(C)}^2, & C\ge C_\ast(d),
\end{cases} \tag{10.5}\label{eq:S-frontier}
\]
where
\[
\lambda(C)\ \text{ is the unique solution of }\ C_{\lambda(C)}=C,
\qquad
C_\ast(d):=C_{\sqrt{d-1}}.
\]
Consequently the optimal fixed-composition risk within this class is
\[
R^{\mathrm{fc}}_{\mathrm{opt}}(C)
=
\frac{d-1}{nd}\left(\frac{1}{S_{\mathrm{opt}}(C)}-1\right). \tag{10.6}\label{eq:mix-opt-risk}
\]
For \(0\le C\le C_\ast(d)\) it is attained by one-level augmented GRR with
\[
\lambda_\ast=\sqrt{d-1},
\qquad
p=\frac{C}{C_\ast(d)}. \tag{10.7}\label{eq:opt-low-budget}
\]
For \(C\ge C_\ast(d)\) it is attained by the calibrated ordinary GRR channel with \(p=1\) and
\(\lambda=\lambda(C)\).
\end{enumerate}
\end{theorem}

\begin{proof}
For \eqref{eq:mix-chi}, fix \(a\neq b\).
In block \(i\), the only nontrivial likelihood-ratio values are
\[
\lambda_i^{-1}\ \text{ on } (i,a),\qquad
\lambda_i\ \text{ on } (i,b),\qquad
1\ \text{ elsewhere},
\]
while every null symbol has ratio \(1\).  Summing the blockwise GRR contributions gives
\eqref{eq:mix-chi}.

For \eqref{eq:mix-risk}, fix \(u\in T_d\), and write
\[
\varphi_u((i,y)):=\frac{\eta_{\lambda_i}}{S}\,u_y,
\qquad
\varphi_u(z):=0,\quad z\in\mathcal Z.
\]
Because \(\sum_{y=1}^d u_y=0\),
\[
\langle u,\widetilde\theta\rangle
=
\frac{1}{n}\sum_{t=1}^n \varphi_u(Y_t).
\]
For an input \(x\),
\[
\E_x[\varphi_u(Y)]
=
\frac{1}{S}\sum_{i=1}^m p_i\eta_{\lambda_i}
\Bigl(\beta_{\lambda_i}\sum_{y=1}^d u_y+\eta_{\lambda_i}u_x\Bigr)
=
u_x
\]
because \(S=\sum_i p_i\eta_{\lambda_i}^2\).
Similarly,
\[
\E_x[\varphi_u(Y)^2]
=
\frac{1}{S^2}\sum_{i=1}^m p_i\eta_{\lambda_i}^2
\Bigl(\beta_{\lambda_i}\|u\|_2^2+\eta_{\lambda_i}u_x^2\Bigr).
\]
Hence
\[
\Var_x(\varphi_u(Y))
=
\frac{1}{S^2}\sum_{i=1}^m p_i\eta_{\lambda_i}^2
\Bigl(\beta_{\lambda_i}\|u\|_2^2+\eta_{\lambda_i}u_x^2\Bigr)-u_x^2.
\]
Averaging over the fixed composition \(\theta\) and writing
\(\Gamma_\theta:=\Cov_\theta(\widetilde\theta)\),
we obtain
\[
u^\top\Gamma_\theta u
=
\frac{1}{n}
\left[
\frac{\sum_i p_i\eta_{\lambda_i}^2\beta_{\lambda_i}}{S^2}\|u\|_2^2
+
\left(
\frac{\sum_i p_i\eta_{\lambda_i}^3}{S^2}-1
\right)
\sum_{x=1}^d \theta_x u_x^2
\right].
\]
Taking traces on \(T_d\) and using
\[
\operatorname{tr}_{T_d}(I)=d-1,
\qquad
\operatorname{tr}_{T_d}\!\bigl(P_T\diag(\theta)P_T\bigr)=\frac{d-1}{d},
\qquad
\beta_{\lambda_i}+\frac{\eta_{\lambda_i}}{d}=\frac1d,
\]
gives
\[
\operatorname{tr}_{T_d}\Gamma_\theta
=
\frac{d-1}{n}
\left[
\frac{1}{S^2}\sum_{i=1}^m p_i\eta_{\lambda_i}^2
\left(\beta_{\lambda_i}+\frac{\eta_{\lambda_i}}{d}\right)
-\frac1d
\right]
=
\frac{d-1}{nd}\left(\frac{1}{S}-1\right).
\]
Since \(\widetilde\theta-\theta\in T_d\), the total squared error equals its \(T_d\)-trace:
\[
\E_\theta\|\widetilde\theta-\theta\|_2^2
=
\operatorname{tr}_{T_d}\Gamma_\theta,
\]
which is exactly \eqref{eq:mix-risk}.

It remains to optimize \(S\) at fixed budget \(C\).
Define
\[
s(\lambda):=\eta_\lambda^2=\frac{(\lambda-1)^2}{(\lambda+d-1)^2},
\qquad
c(\lambda):=C_\lambda=\frac{(\lambda-1)^2(\lambda+1)}{\lambda(\lambda+d-1)}.
\]
Since
\[
c'(\lambda)
=
\frac{(\lambda-1)\bigl(2d\lambda^2+d\lambda+d+\lambda^3-\lambda^2+\lambda-1\bigr)}
{\lambda^2(\lambda+d-1)^2}
>0,
\]
the equation \(c(\lambda)=C\) has a unique solution \(\lambda(C)\) for every \(C>0\).

Every channel in the theorem yields a convex combination
\[
(C,S)=\sum_{i=1}^m p_i\bigl(c(\lambda_i),s(\lambda_i)\bigr)+(1-\sum_i p_i)(0,0).
\]
Hence the optimal frontier is the upper concave envelope of the planar curve
\(\lambda\mapsto (c(\lambda),s(\lambda))\) together with the origin.

First,
\[
\frac{s(\lambda)}{c(\lambda)}
=
\frac{\lambda}{(\lambda+1)(\lambda+d-1)}. \tag{10.8}\label{eq:ratio-sc}
\]
Its derivative is
\[
\frac{d}{d\lambda}\frac{s(\lambda)}{c(\lambda)}
=
\frac{d-1-\lambda^2}{(\lambda+1)^2(\lambda+d-1)^2}, \tag{10.9}\label{eq:ratio-derivative}
\]
so the slope from the origin is maximized uniquely at
\(
\lambda_\ast=\sqrt{d-1}
\).
Thus the line from the origin tangent to the curve meets it at
\[
\bigl(c(\lambda_\ast),s(\lambda_\ast)\bigr)
=
\left(
C_\ast(d),\frac{C_\ast(d)}{d+2\sqrt{d-1}}
\right).
\]

Second,
\[
\frac{ds}{dc}
=
\frac{2d\lambda^2}{(\lambda+d-1)\bigl(2d\lambda^2+d\lambda+d+\lambda^3-\lambda^2+\lambda-1\bigr)}, \tag{10.10}\label{eq:dsdc}
\]
and therefore
\[
\frac{d^2s}{dc^2}
=
\frac{2d\lambda^3 P_d(\lambda)}
{(\lambda-1)\bigl(2d\lambda^2+d\lambda+d+\lambda^3-\lambda^2+\lambda-1\bigr)^3}, \tag{10.11}\label{eq:d2sdc2}
\]
where
\[
P_d(\lambda)
=
d^2\lambda+2d^2-3d\lambda^3+d\lambda-4d-2\lambda^4+2\lambda^3-2\lambda+2.
\]
For \(\lambda\ge \sqrt{d-1}\),
\[
P_d'(\lambda)
=
d^2+d-2-(9d-6)\lambda^2-8\lambda^3
\le
-8(d-1)^2-8(d-1)^{3/2}<0, \tag{10.12}\label{eq:P-decreasing}
\]
and
\[
P_d(\sqrt{d-1})=2(d-1)^{3/2}(2-d)\le 0. \tag{10.13}\label{eq:P-at-star}
\]
Hence \(P_d(\lambda)\le 0\) for every \(\lambda\ge \sqrt{d-1}\), so
\[
\frac{d^2s}{dc^2}\le 0
\qquad\text{for } \lambda\ge \sqrt{d-1}. \tag{10.14}\label{eq:curve-concave}
\]
Thus the upper concave envelope consists exactly of the tangent line from the origin up to
\(C_\ast(d)\), followed by the original curve for \(C\ge C_\ast(d)\).
This proves \eqref{eq:S-frontier}, and \eqref{eq:mix-opt-risk} follows from \eqref{eq:mix-risk}.
The descriptions of the extremizers are immediate from the two pieces of the concave envelope.
\end{proof}

\begin{corollary}[Augmented GRR beats calibrated GRR at low budget]\label{cor:aug-better}
Let \(0<C<C_\ast(d)\), and let \(\lambda(C)\) be the unique solution of \(C_{\lambda(C)}=C\).
Then the optimal projected inverse-estimator risk is
\[
R^{\mathrm{fc}}_{\mathrm{opt}}(C)
=
\frac{d-1}{nd}\left(\frac{d+2\sqrt{d-1}}{C}-1\right), \tag{10.15}\label{eq:aug-risk-low}
\]
whereas calibrated GRR has projected risk
\[
R^{\mathrm{fc}}_{\mathrm{GRR}}(C)
=
\frac{d-1}{nd}\left(\frac{d+\lambda(C)+(d-1)/\lambda(C)}{C}-1\right). \tag{10.16}\label{eq:grr-risk-cal}
\]
Hence
\[
\frac{R^{\mathrm{fc}}_{\mathrm{opt}}(C)}{R^{\mathrm{fc}}_{\mathrm{GRR}}(C)}
=
\frac{d+2\sqrt{d-1}-C}{d+\lambda(C)+(d-1)/\lambda(C)-C}
<1. \tag{10.17}\label{eq:improvement-factor}
\]
\end{corollary}

\begin{example}[Two numerical instances]\label{ex:aug-grr-numbers}
For \(d=3\) and \(C=0.05\),
\[
C_\ast(3)=\tfrac{3-2\sqrt2}{2}\approx 0.085786,
\qquad
p=\frac{0.05}{C_\ast(3)}\approx 0.582843,
\]
and
\[
R^{\mathrm{fc}}_{\mathrm{opt}}(0.05)\approx \frac{77.0457}{n},
\qquad
R^{\mathrm{fc}}_{\mathrm{GRR}}(0.05)\approx \frac{77.1653}{n}.
\]
For \(d=10\) and \(C=0.1\),
\[
C_\ast(10)=\frac49,
\qquad
p=\frac{0.1}{4/9}=0.225,
\]
and
\[
R^{\mathrm{fc}}_{\mathrm{opt}}(0.1)=\frac{143.1}{n},
\qquad
R^{\mathrm{fc}}_{\mathrm{GRR}}(0.1)\approx \frac{149.7150}{n}.
\]
\end{example}

\begin{remark}[Null refinements]\label{rem:null-refinement}
The theorem depends on the null part only through its total mass.
In particular, replacing the single null symbol \(z\) by any finite family
\(z_1,\dots,z_k\) with probabilities \((1-p)/k\) leaves both the canonical budget and the
projected fixed-composition risk unchanged.
\end{remark}

\begin{remark}[Scope of Section~\ref{sec:thinning}]\label{rem:not-all-equivariant}
The class of Theorem~\ref{thm:grr-mixtures} consists of disjoint unions of copies of \([d]\)
together with null orbits.
Section~\ref{sec:equivariant-opt} shows that in the low-budget regime \(0<C\le C_\ast(d)\)
these already attain the full permutation-equivariant optimum.
What remains open is the global symmetric frontier for \(C>C_\ast(d)\), where ordered-pair
and higher orbits may in principle matter.
\end{remark}

\section{Low-budget optimality among all permutation-equivariant channels}\label{sec:equivariant-opt}

We now allow an arbitrary finite output alphabet \(\mathcal Y\) equipped with an action of \(S_d\),
and we assume throughout this section that
\[
W(\pi y\mid \pi x)=W(y\mid x)
\qquad
\text{for every } \pi\in S_d,\ x\in[d],\ y\in\mathcal Y.
\]
Write the orbit decomposition as
\(
\mathcal Y=\bigsqcup_{\mathcal O}\mathcal O
\),
and let
\[
\mu(y):=\frac1d\sum_{x=1}^d W(y\mid x)
\]
be the row-average output law.
By equivariance, \(\mu\) is \(S_d\)-invariant, hence constant on each orbit.

For every orbit \(\mathcal O\), choose a representative \(y_{\mathcal O}\in\mathcal O\),
let \(p_{\mathcal O}:=\mu(\mathcal O)\), and define the orbit template
\[
a^{(\mathcal O)}_x
:=
\frac{W(y_{\mathcal O}\mid x)}{\mu(y_{\mathcal O})},
\qquad x=1,\dots,d.
\]
If the canonical pairwise \(\chi^2\)-budget is finite, then every entry of \(a^{(\mathcal O)}\) is strictly
positive, and
\[
\sum_{x=1}^d a^{(\mathcal O)}_x=d. \tag{11.1}\label{eq:orbit-template-sum}
\]
Write
\[
t^{(\mathcal O)}:=a^{(\mathcal O)}-\one\in T_d,
\qquad
A_{\mathcal O}:=\sum_{x=1}^d \frac{1}{a_x^{(\mathcal O)}},
\qquad
B_{\mathcal O}:=\sum_{x=1}^d \bigl(a_x^{(\mathcal O)}\bigr)^2.
\]
Every symbol in \(\mathcal O\) carries a permutation of this template.

\begin{proposition}[Orbit templates, projected estimator, and additive invariants]\label{prop:orbit-template}
For a permutation-equivariant channel \(W\) with finite output alphabet, define
\[
S(W):=\sum_{\mathcal O} p_{\mathcal O}\,\frac{B_{\mathcal O}-d}{d(d-1)}. \tag{11.2}\label{eq:orbit-signal}
\]
Then:

\begin{enumerate}[label=\textup{(\roman*)}]
\item For every pair \(a\neq b\),
\[
\chi^2\!\bigl(W(\cdot\mid b)\,\|\,W(\cdot\mid a)\bigr)
=
\sum_{\mathcal O}
p_{\mathcal O}\,
\frac{A_{\mathcal O}B_{\mathcal O}-d^2}{d(d-1)}. \tag{11.3}\label{eq:orbit-budget}
\]

\item Assume \(S(W)>0\) (equivalently, at least one orbit is informative).
Let \(t(y)\in T_d\) denote the score vector attached to \(y\), obtained by permuting the relevant
template \(t^{(\mathcal O)}\) on the orbit containing \(y\).
Then the estimator
\[
\widetilde\theta
=
\frac{\one}{d}
+
\frac{1}{nd\,S(W)}
\sum_{y\in\mathcal Y} N_y\,t(y) \tag{11.4}\label{eq:orbit-estimator}
\]
is unbiased, satisfies \(\one^\top\widetilde\theta=1\), and has exact fixed-composition risk
\[
\E_\theta\|\widetilde\theta-\theta\|_2^2
=
\frac{d-1}{nd}\left(\frac{1}{S(W)}-1\right), \tag{11.5}\label{eq:orbit-risk}
\]
constant in \(\theta\).

\item If \(W\) is supported on a single informative orbit with template \(a\) (so \(B>d\)), then
\[
\frac{S(W)}{\chi^2(W(\cdot\mid b)\,\|\,W(\cdot\mid a))}
=
\frac{B-d}{AB-d^2}. \tag{11.6}\label{eq:orbit-slope}
\]
\end{enumerate}
\end{proposition}

\begin{proof}
Fix an orbit \(\mathcal O\), and let \(m_{\mathcal O}:=|\mathcal O|\).
Since \(\mu\) is uniform on \(\mathcal O\), we have
\(
\mu(y_{\mathcal O})=p_{\mathcal O}/m_{\mathcal O}
\).
For \(y=\pi y_{\mathcal O}\in\mathcal O\), equivariance gives
\[
W(y\mid x)
=
\frac{p_{\mathcal O}}{m_{\mathcal O}}\,
a^{(\mathcal O)}_{\pi^{-1}x}
=
\frac{p_{\mathcal O}}{m_{\mathcal O}}\,
\bigl(1+t^{(\mathcal O)}_{\pi^{-1}x}\bigr). \tag{11.7}\label{eq:orbit-kernel}
\]
Hence each row places total mass \(p_{\mathcal O}\) on \(\mathcal O\), independently of \(x\).

For \eqref{eq:orbit-budget}, fix the pair \((1,2)\).
Using \eqref{eq:orbit-kernel}, the orbitwise contribution to the \(\chi^2\)-budget is
\[
\frac{p_{\mathcal O}}{m_{\mathcal O}}
\sum_{y\in\mathcal O}
\frac{\bigl(t_2(y)-t_1(y)\bigr)^2}{1+t_1(y)}.
\]
Averaging uniformly over the orbit is the same as averaging over all \(d!\) permutations of the template,
because each distinct orbit point arises the same number of times from labeled permutations.
Therefore the ordered pair of template coordinates seen at \((1,2)\) is uniformly distributed over all
ordered pairs \((i,j)\) with \(i\neq j\), and the orbit contribution equals
\[
p_{\mathcal O}\,
\frac{1}{d(d-1)}
\sum_{i\neq j}
\frac{(a_j^{(\mathcal O)}-a_i^{(\mathcal O)})^2}{a_i^{(\mathcal O)}}.
\]
Now
\[
\sum_{j\neq i}(a_j-a_i)^2
=
\sum_{j=1}^d a_j^2-2a_i\sum_{j=1}^d a_j+d a_i^2
=
B-2d a_i+d a_i^2
\]
by \eqref{eq:orbit-template-sum}.
Summing over \(i\) gives exactly \eqref{eq:orbit-budget}.

For \eqref{eq:orbit-estimator}--\eqref{eq:orbit-risk}, define
\[
M:=\sum_{y\in\mathcal Y}\mu(y)\,t(y)t(y)^\top.
\]
Because \(t(\pi y)=P_\pi t(y)\) and \(\mu\) is invariant, \(M\) commutes with every permutation matrix.
Hence
\(
M=\alpha P_T
\)
for some \(\alpha\ge 0\).
Taking traces on \(T_d\) gives
\[
\alpha
=
\frac{1}{d-1}\sum_{y\in\mathcal Y}\mu(y)\|t(y)\|_2^2
=
\frac{1}{d-1}\sum_{\mathcal O} p_{\mathcal O}\|t^{(\mathcal O)}\|_2^2
=
d\,S(W), \tag{11.8}\label{eq:orbit-matrix}
\]
so
\[
M=d\,S(W)\,P_T.
\]
For an input \(x\),
\[
\begin{aligned}
\E_x[t(Y)]
&=
\sum_{y\in\mathcal Y}\mu(y)\bigl(1+t_x(y)\bigr)t(y)\\
&=
\sum_{y\in\mathcal Y}\mu(y)t_x(y)t(y)
=
M e_x\\
&=
d\,S(W)\left(e_x-\frac{\one}{d}\right),
\end{aligned}
\]
where the second equality uses \(\sum_{y\in\mathcal Y}\mu(y)t(y)=0\), which holds because
\(\mu\) is \(S_d\)-invariant and \(t(y)\in T_d\), so the sum is an \(S_d\)-invariant vector in \(T_d\)
and must vanish.
Therefore \(\widetilde\theta\) in \eqref{eq:orbit-estimator} is unbiased,
and it takes values in the affine hyperplane
\(
\{\vartheta:\one^\top\vartheta=1\}.
\)

Next, \(\|t(y)\|_2=\|t^{(\mathcal O)}\|_2\) on each orbit \(\mathcal O\), and each row assigns total mass
\(p_{\mathcal O}\) to that orbit, so for every input \(x\),
\[
\E_x\|t(Y)\|_2^2
=
\sum_{\mathcal O} p_{\mathcal O}\,\|t^{(\mathcal O)}\|_2^2
=
d(d-1)\,S(W). \tag{11.9}\label{eq:orbit-second-moment}
\]
The mean contribution of one user with input \(x\) is
\(
d\,S(W)(e_x-\one/d)
\),
whose squared norm is
\[
d^2 S(W)^2
\left\|e_x-\frac{\one}{d}\right\|_2^2
=
d(d-1)\,S(W)^2.
\]
Since users are independent under fixed composition,
\[
\begin{aligned}
\E_\theta\|\widetilde\theta-\theta\|_2^2
&=
\frac{1}{nd^2 S(W)^2}
\left(
d(d-1)\,S(W)-d(d-1)\,S(W)^2
\right)\\
&=
\frac{d-1}{nd}\left(\frac{1}{S(W)}-1\right),
\end{aligned}
\]
which proves \eqref{eq:orbit-risk}.
Finally, \eqref{eq:orbit-slope} is just the quotient of \eqref{eq:orbit-signal} and \eqref{eq:orbit-budget}
for a single orbit.
\end{proof}

\begin{lemma}[Two-level reduction at fixed second moment]\label{lem:two-level-template}
Fix \(d\ge 2\) and \(B>d\).  Suppose that
\[
\mathcal K_B
:=
\left\{
a\in(0,\infty)^d:
\sum_{i=1}^d a_i=d,\quad
\sum_{i=1}^d a_i^2=B
\right\}
\]
is nonempty.  Then the functional
\[
A(a):=\sum_{i=1}^d \frac{1}{a_i}
\]
attains its minimum on \(\mathcal K_B\), and every minimizer has at most two
distinct coordinate values.  Consequently, for every \(a\in\mathcal K_B\),
\[
\sum_{i=1}^d \frac{1}{a_i}
\ge
\frac{s}{\alpha}+\frac{d-s}{\beta}
\]
for some \(s\in\{1,\dots,d-1\}\) and some \(\alpha>\beta>0\) satisfying
\[
s\alpha+(d-s)\beta=d,
\qquad
s\alpha^2+(d-s)\beta^2=B.
\]
\end{lemma}

\begin{proof}
Let \(a^{(0)}\in\mathcal K_B\) and set \(M_0=A(a^{(0)})\).  Take a minimizing
sequence \(a^{(m)}\in\mathcal K_B\) with \(A(a^{(m)})\le M_0+1\) eventually.
Since every summand \(1/a^{(m)}_i\) is positive, we have
\[
a^{(m)}_i\ge \frac{1}{M_0+1}
\qquad
\text{for every }i
\]
along the tail of the sequence.  Also \(a^{(m)}_i\le d\), because
\(\sum_i a^{(m)}_i=d\).  Hence a subsequence converges to some
\(a^\star\in[1/(M_0+1),d]^d\).  Passing to the limit in the two constraints
shows \(a^\star\in\mathcal K_B\), and continuity of \(A\) on this compact
box shows that \(a^\star\) is a minimizer.

Since \(B>d\), the minimizer is not the constant vector \(\one\).  Therefore
the gradients of the two equality constraints,
\(
\nabla(\sum_i a_i)=\one
\)
and
\(
\nabla(\sum_i a_i^2)=2a^\star
\),
are linearly independent at \(a^\star\).  The equality-constrained KKT
conditions are therefore necessary.  Thus there exist real multipliers
\(\alpha_0,\beta_0\) such that
\[
-\frac{1}{(a_i^\star)^2}+\alpha_0+2\beta_0 a_i^\star=0,
\qquad i=1,\dots,d.
\]
Equivalently, every coordinate \(a_i^\star\) is a positive zero of
\[
q(x):=\alpha_0+2\beta_0 x-x^{-2},
\qquad x>0.
\]

If \(\beta_0\ge 0\), then \(q'(x)=2\beta_0+2x^{-3}>0\) for every \(x>0\), so
\(q\) has at most one positive zero.  This would force all coordinates of
\(a^\star\) to be equal, impossible when \(B>d\).
Thus the nontrivial case has \(\beta_0<0\).  Then \(q'\) has exactly one
positive zero, namely \(x=(-1/\beta_0)^{1/3}\), and \(q\) is increasing before
this point and decreasing after it.  Hence \(q\) has at most two positive
zeros.  Therefore the coordinates of \(a^\star\) take at most two distinct
positive values.

Because \(B>d\), the one-level case is impossible.  Writing the two values as
\(\alpha>\beta>0\), with multiplicities \(s\) and \(d-s\), gives
\(s\in\{1,\dots,d-1\}\) and the two displayed constraints.  Since
\(a^\star\) minimizes \(A\) over \(\mathcal K_B\), the final inequality follows.
\end{proof}

\begin{theorem}[Low-budget optimality among all permutation-equivariant channels]\label{thm:equivariant-opt}
Assume \(d\ge 3\) and \(0<C\le C_\ast(d)\).
Among all permutation-equivariant channels \(W:[d]\to\Delta(\mathcal Y)\) with arbitrary finite output
alphabet and canonical pairwise \(\chi^2\)-budget \(C\), the minimum exact fixed-composition risk of the
projected unbiased inverse estimator is
\[
R^{\mathrm{fc}}_{\mathrm{sym}}(C)
=
\frac{d-1}{nd}\left(\frac{d+2\sqrt{d-1}}{C}-1\right). \tag{11.10}\label{eq:equivariant-opt-risk}
\]
It is attained by augmented GRR with
\[
\lambda_\ast=\sqrt{d-1},
\qquad
p=\frac{C}{C_\ast(d)}. \tag{11.11}\label{eq:equivariant-opt-channel}
\]
More generally, equality holds if and only if every informative orbit is, up to a permutation of the input
labels, the singleton GRR template
\[
a_\ast
=
\left(
\frac{d\sqrt{d-1}}{\sqrt{d-1}+d-1},
\frac{d}{\sqrt{d-1}+d-1},
\dots,
\frac{d}{\sqrt{d-1}+d-1}
\right), \tag{11.12}\label{eq:grr-template}
\]
and the remaining mass lies on null refinements.
Equivalently, every optimizer is an augmented-GRR channel up to splitting the informative GRR symbols
and/or the null symbol into finitely many conditionally identical refinements.
\end{theorem}

\begin{proof}
By Proposition~\ref{prop:orbit-template}, minimizing the projected risk is equivalent to maximizing the
total signal coefficient \(S(W)\) at fixed total budget \(C\).
Because both \(S(W)\) and the budget in \eqref{eq:orbit-budget} add over orbits, it suffices to bound the
orbitwise slope \eqref{eq:orbit-slope}.

Fix one informative orbit and abbreviate \(A=A_{\mathcal O}\), \(B=B_{\mathcal O}\), \(a=a^{(\mathcal O)}\).
Since the orbit is informative, \(B>d\).  For fixed \(B\), the numerator \(B-d\) in
\eqref{eq:orbit-slope} is fixed, while the denominator \(AB-d^2\) is strictly increasing in \(A\).
Thus the orbitwise slope is maximized, at fixed \(B\), when
\[
A=\sum_{i=1}^d \frac{1}{a_i}
\]
is minimized subject to
\(
a_i>0
\),
\(
\sum_i a_i=d
\),
and
\(
\sum_i a_i^2=B
\).
By Lemma~\ref{lem:two-level-template}, every such minimizer is two-level.  Hence it has the form
\(
a_i\in\{\alpha,\beta\}
\),
\(
\alpha>\beta>0
\),
with multiplicities \(s\) and \(d-s\), where \(s\in\{1,\dots,d-1\}\).  Since
\(s\alpha+(d-s)\beta=d\), setting \(\lambda:=\alpha/\beta>1\) gives
\[
\alpha=\frac{d\lambda}{d+s(\lambda-1)},
\qquad
\beta=\frac{d}{d+s(\lambda-1)}.
\]
This is the subset-selection template of size \(s\), possibly with conditionally identical refinements
of the corresponding orbit.
A direct substitution into \eqref{eq:orbit-signal} and \eqref{eq:orbit-budget} yields
\[
S_s(\lambda)
=
\frac{s(d-s)(\lambda-1)^2}{(d-1)\bigl(d+s(\lambda-1)\bigr)^2},
\qquad
C_s(\lambda)
=
\frac{s(d-s)(\lambda-1)^2(\lambda+1)}
{\lambda(d-1)\bigl(d+s(\lambda-1)\bigr)}, \tag{11.13}\label{eq:two-level-sc}
\]
(we write \(C_s\) for \(C_{d,s}\) when \(d\) is fixed),
and therefore
\[
\frac{S_s(\lambda)}{C_s(\lambda)}
=
\frac{\lambda}{(\lambda+1)\bigl(d+s(\lambda-1)\bigr)}. \tag{11.14}\label{eq:two-level-slope}
\]
For fixed \(s\), differentiation gives
\[
\frac{d}{d\lambda}\frac{S_s(\lambda)}{C_s(\lambda)}
=
\frac{d-s-s\lambda^2}{(\lambda+1)^2\bigl(d+s(\lambda-1)\bigr)^2}. \tag{11.15}\label{eq:two-level-derivative}
\]
If \(1\le s< d/2\), the unique critical point on \((1,\infty)\) is
\(
\lambda_s^\ast=\sqrt{(d-s)/s}>1
\),
and
\[
\max_{\lambda>1}\frac{S_s(\lambda)}{C_s(\lambda)}
=
\frac{1}{d+2\sqrt{s(d-s)}}. \tag{11.16}\label{eq:two-level-max}
\]
This maximum is largest when \(s=1\), since \(s(d-s)\ge d-1\) for \(1\le s\le d-1\).

If \(s\ge d/2\), then \(d-s-s\lambda^2\le 0\) for every \(\lambda\ge 1\), so the slope in
\eqref{eq:two-level-slope} is non-increasing on \([1,\infty)\) and its supremum is the boundary value
\(1/(2d)\), which is strictly smaller than \(1/(d+2\sqrt{d-1})\) for \(d\ge 3\).
Thus every informative orbit satisfies
\[
\frac{S_{\mathcal O}}{C_{\mathcal O}}
\le
\frac{1}{d+2\sqrt{d-1}}, \tag{11.17}\label{eq:orbit-global-slope}
\]
with equality if and only if \(s=1\) and \(\lambda=\sqrt{d-1}\).

Summing \eqref{eq:orbit-global-slope} over the orbits and using Proposition~\ref{prop:orbit-template} gives
\[
S(W)\le \frac{C}{d+2\sqrt{d-1}}.
\]
Hence every permutation-equivariant channel with budget \(C\) obeys
\[
\E_\theta\|\widetilde\theta-\theta\|_2^2
\ge
\frac{d-1}{nd}\left(\frac{d+2\sqrt{d-1}}{C}-1\right).
\]
This is \eqref{eq:equivariant-opt-risk}.

Finally, augmented GRR with \(\lambda_\ast=\sqrt{d-1}\) and \(p=C/C_\ast(d)\) is feasible when
\(0<C\le C_\ast(d)\), and has
\[
S(W)
=
p\,\eta_{\lambda_\ast}^2
=
\frac{C}{d+2\sqrt{d-1}},
\]
so it attains \eqref{eq:equivariant-opt-risk}.
Equality throughout forces every informative orbit to attain equality in
\eqref{eq:orbit-global-slope}, hence to be the singleton template \eqref{eq:grr-template} up to a
permutation of the coordinates.
All remaining mass must lie on null orbits.
This is exactly the stated refinement-class characterization.
\end{proof}

\begin{remark}[Ordered-pair orbits]\label{rem:ordered-pairs}
The ordered-pair orbit \(\{(i,j):i\neq j\}\) corresponds to a three-level template
\(
a=(u,v,w,\dots,w)
\)
with \(u+v+(d-2)w=d\).
Here
\[
\begin{aligned}
S_{\mathrm{op}}
&=
\frac{u^2+v^2+(d-2)w^2-d}{d(d-1)},\\
C_{\mathrm{op}}
&=
\frac{\bigl(1/u+1/v+(d-2)/w\bigr)\bigl(u^2+v^2+(d-2)w^2\bigr)-d^2}{d(d-1)}.
\end{aligned}
\]
Thus the ordered-pair slope is exactly \(S_{\mathrm{op}}/C_{\mathrm{op}}=(B-d)/(AB-d^2)\) with
\(A=1/u+1/v+(d-2)/w\) and \(B=u^2+v^2+(d-2)w^2\).
The proof of Theorem~\ref{thm:equivariant-opt} shows that, at fixed \(B\), the minimizing \(A\) must be
two-level.
Hence an ordered-pair orbit can only attain the optimal slope on the boundary \(u=w\) or \(v=w\),
where it collapses to a refined GRR orbit.
In particular, ordered pairs never improve on augmented GRR in the low-budget regime.
\end{remark}

\section{GRR optimal within the subset-selection family}\label{sec:ss}

For \(1\le s\le d-1\) and \(\lambda>1\), let \(SS(d,s,\lambda)\) denote the subset-selection channel on
\(\binom{[d]}{s}\):
\[
W_{d,s,\lambda}(S\mid x)
=
\frac{d}{\binom ds\,(\lambda s+d-s)}
\Bigl(\lambda\,\one\{x\in S\}+\one\{x\notin S\}\Bigr). \tag{12.1}\label{eq:ss-kernel}
\]
Write
\[
C_{d,s}(\lambda)
:=
\frac{s(d-s)(\lambda-1)^2(\lambda+1)}
{\lambda(d-1)(\lambda s+d-s)} \tag{12.2}\label{eq:ss-budget}
\]
for the canonical pairwise \(\chi^2\) divergence, and
\[
R^{\mathrm{fc}}_{d,s}(\lambda)
:=
\frac{d-1}{n}\,
\frac{\lambda^2s(s-1)+2\lambda s(d-s)+(d-s)(d-s-1)}
{s(d-s)(\lambda-1)^2} \tag{12.3}\label{eq:ss-risk}
\]
for the exact projected fixed-composition risk.

Formula \eqref{eq:ss-budget} is immediate from a four-type count.
Fix \(a\neq b\).
Subsets \(S\in\binom{[d]}{s}\) split into the four classes
\[
a\in S,\ b\notin S;\qquad
b\in S,\ a\notin S;\qquad
a,b\in S;\qquad
a,b\notin S.
\]
The first two classes each have cardinality \(\binom{d-2}{s-1}\).
Under input \(a\), their likelihood-ratio values relative to input \(b\) are
\(\lambda^{-1}\) and \(\lambda\), respectively, while the last two classes are neutral.
Using
\[
\frac{\binom{d-2}{s-1}}{\binom ds}
=
\frac{s(d-s)}{d(d-1)},
\]
one obtains \eqref{eq:ss-budget}.

For the inverse estimator, let
\[
C_j:=\sum_{i=1}^n \one\{j\in S_i\},
\qquad j=1,\dots,d,
\]
be the coordinate-inclusion counts.
A direct count gives
\[
\begin{aligned}
p_s&:=\PP_x(x\in S)=\frac{\lambda s}{d+s(\lambda-1)},\\
r_s&:=\PP_x(j\in S\mid j\neq x)
=
\frac{s\bigl(\lambda(s-1)+d-s\bigr)}
{(d-1)\bigl(d+s(\lambda-1)\bigr)}.
\end{aligned} \tag{12.4}\label{eq:ss-pr}
\]
Hence
\[
\E_\theta\!\left[\frac{C_j}{n}\right]
=
r_s+(p_s-r_s)\theta_j,
\qquad
p_s-r_s
=
\frac{s(d-s)(\lambda-1)}
{(d-1)\bigl(d+s(\lambda-1)\bigr)}. \tag{12.5}\label{eq:ss-gap}
\]
Therefore
\[
\widehat\theta_j
:=
\frac{C_j/n-r_s}{p_s-r_s},
\qquad j=1,\dots,d, \tag{12.6}\label{eq:ss-inverse}
\]
is unbiased, and it already satisfies \(\sum_j \widehat\theta_j=1\) because \(\sum_j C_j=ns\).

Equivalently, \(SS(d,s,\lambda)\) is the two-level orbit template from
Section~\ref{sec:equivariant-opt} with
\[
\alpha=\frac{d\lambda}{d+s(\lambda-1)},
\qquad
\beta=\frac{d}{d+s(\lambda-1)},
\]
where \(\alpha\) appears on \(s\) coordinates and \(\beta\) on the remaining \(d-s\) coordinates.
Substituting this template into Proposition~\ref{prop:orbit-template} yields
\eqref{eq:ss-risk}, and the estimator \eqref{eq:ss-inverse} is exactly the projected orbit estimator
\eqref{eq:orbit-estimator} written in inclusion-count coordinates.

\begin{theorem}[GRR is the unique optimizer in the subset-selection family]\label{thm:ss-opt}
Fix \(d\ge 2\) and \(C>0\).
For each \(s\in\{1,\dots,d-1\}\), let \(\lambda_s(C)>1\) be the unique solution of
\[
C_{d,s}\bigl(\lambda_s(C)\bigr)=C. \tag{12.7}\label{eq:ss-lam}
\]
Define the matched-budget risk
\[
\mathfrak R_{d,s}(C):=R^{\mathrm{fc}}_{d,s}\bigl(\lambda_s(C)\bigr). \tag{12.8}\label{eq:ss-matched-risk}
\]
Then
\[
\mathfrak R_{d,1}(C)<\mathfrak R_{d,2}(C)<\cdots<\mathfrak R_{d,d-1}(C). \tag{12.9}\label{eq:ss-monotone}
\]
Equivalently, GRR (\(s=1\)) is the unique matched-budget minimizer in the entire subset-selection family.
\end{theorem}

\begin{proof}
If \(d=2\), then \(s=1\) is the only admissible subset size, so the claim is immediate.
Assume henceforth that \(d\ge 3\).
Extend \(C_{d,s}(\lambda)\) and \(R^{\mathrm{fc}}_{d,s}(\lambda)\) to real \(s\in[1,d-1]\) by the same formulas.
For fixed \(s\in(1,d-1)\) and \(\lambda>1\),
\[
\partial_\lambda C_{d,s}(\lambda)
=
\frac{s(d-s)(\lambda-1)\bigl(2d\lambda^2+d\lambda+d+\lambda^3 s-\lambda^2 s+\lambda s-s\bigr)}
{\lambda^2(d-1)\bigl(d+s(\lambda-1)\bigr)^2}
>0. \tag{12.10}\label{eq:ss-budget-derivative}
\]
Hence, for every fixed \(C>0\), the equation \(C_{d,s}(\lambda)=C\) has a unique solution
\(\lambda_s(C)>1\), and by the implicit function theorem the map
\(s\mapsto \lambda_s(C)\) is \(C^1\) on \((1,d-1)\).
Since monotonicity on the real interval \([1,d-1]\) implies monotonicity on the integer points,
it suffices to differentiate the real-valued function \(s\mapsto \mathfrak R_{d,s}(C)\).

Using \eqref{eq:ss-budget} and \eqref{eq:ss-risk},
\[
\mathfrak R_{d,s}(C)
=
\frac{d-1}{ndC}
\left(
\frac{(\lambda_s(C)+1)\bigl(d+s(\lambda_s(C)-1)\bigr)}{\lambda_s(C)}-C
\right). \tag{12.11}\label{eq:ss-risk-simple}
\]
Write \(\lambda=\lambda_s(C)\) for brevity.
Since \(C_{d,s}(\lambda)=C\), implicit differentiation gives
\[
\frac{d\lambda}{ds}
=
-\frac{\partial_s C_{d,s}(\lambda)}{\partial_\lambda C_{d,s}(\lambda)}
=
-\frac{\lambda(\lambda-1)(\lambda+1)\bigl(d^2-2ds-s^2(\lambda-1)\bigr)}
{s(d-s)\bigl(2d\lambda^2+d\lambda+d+\lambda^3 s-\lambda^2 s+\lambda s-s\bigr)}. \tag{12.12}\label{eq:implicit-lambda}
\]
Also,
\[
\begin{aligned}
\partial_s\!\left(
\frac{(\lambda+1)(d+s(\lambda-1))}{\lambda}
\right)
&=
\frac{\lambda^2-1}{\lambda},\\
\partial_\lambda\!\left(
\frac{(\lambda+1)(d+s(\lambda-1))}{\lambda}
\right)
&=
\frac{\lambda^2 s+s-d}{\lambda^2}.
\end{aligned} \tag{12.13}\label{eq:ss-partials}
\]
Substituting \eqref{eq:implicit-lambda} and \eqref{eq:ss-partials} into the derivative of
\eqref{eq:ss-risk-simple} and simplifying yields
\[
\frac{\partial}{\partial s}\mathfrak R_{d,s}(C)
=
\frac{d-1}{nC}\,
\frac{
(\lambda-1)(\lambda+1)\bigl(d+\lambda s-s\bigr)\bigl(d+\lambda^2 s-s\bigr)
}{
\lambda s(d-s)\bigl(2d\lambda^2+d\lambda+d+\lambda^3 s-\lambda^2 s+\lambda s-s\bigr)
}. \tag{12.14}\label{eq:ss-derivative}
\]
The denominator satisfies
\[
2d\lambda^2+d\lambda+d+\lambda^3 s-\lambda^2 s+\lambda s-s
=
2d\lambda^2+d\lambda+d+s(\lambda-1)(\lambda^2+1)>0. \tag{12.15}\label{eq:ss-den-positive}
\]
Every other factor in \eqref{eq:ss-derivative} is also strictly positive for
\(1<s<d-1\) and \(\lambda>1\).
Hence \(s\mapsto \mathfrak R_{d,s}(C)\) is strictly increasing on \((1,d-1)\).
Since the formulas \eqref{eq:ss-budget} and \eqref{eq:ss-risk} extend continuously to real \(s\in[1,d-1]\)
and \(\lambda_s(C)\) is smooth on this interval (by the implicit function theorem and
\eqref{eq:ss-budget-derivative}), strict monotonicity on the open interval implies
\(\mathfrak R_{d,1}(C)<\mathfrak R_{d,s}(C)\) for every integer \(s\ge 2\),
proving \eqref{eq:ss-monotone}.
\end{proof}

\begin{remark}[Thinned subset selection]\label{rem:ss-thinned}
If one allows an additional null symbol and puts only a fraction \(p\) of the mass on
\(SS(d,s,\lambda)\), then Proposition~\ref{prop:orbit-template} shows that the exact
fixed-composition risk of the affine-projected inverse estimator has the form
\[
\frac{d-1}{nd}\left(\frac{1}{p\kappa_{d,s}(\lambda)}-1\right),
\qquad
\kappa_{d,s}(\lambda)=
\frac{s(d-s)(\lambda-1)^2}{(d-1)\bigl(d+s(\lambda-1)\bigr)^2},
\]
while the canonical budget becomes \(p\,C_{d,s}(\lambda)\).
The ratio
\[
\frac{\kappa_{d,s}(\lambda)}{C_{d,s}(\lambda)}
=
\frac{\lambda}{(\lambda+1)\bigl(d+s(\lambda-1)\bigr)}
\]
is maximized at \(s=1\) and \(\lambda=\sqrt{d-1}\).
Thus even after adding null mass, the subset-selection family never beats the low-budget
augmented-GRR frontier of Sections~\ref{sec:thinning} and~\ref{sec:equivariant-opt}.
\end{remark}

\section{Discussion}\label{sec:new-discussion}

The results above separate four layers.

First, Theorem~\ref{thm:near-vertex-cr} and Theorem~\ref{thm:assouad-risk} show that the worst pairwise
canonical quantity
\(
\chi_\ast(W)
\)
already forces a universal estimation lower bound of order
\(
(d-1)/(n\chi_\ast(W))
\).
Thus \(\chi_\ast\) is not only the canonical privacy invariant of the present paper; it is also a universal
statistical obstruction.

Second, Theorem~\ref{thm:symmetrization} shows that, for the uniform-point Fisher criterion, asymmetry never
helps.  One may average over \(S_d\), decrease the pairwise \(\chi^2\) budget, and equalize the Fisher
eigenvalues.

Third, Theorem~\ref{thm:grr-mixtures} identifies the exact optimal frontier inside the natural class of GRR
blocks and null refinements.  The optimal projected inverse-estimator risk is attained by the thinned GRR
mechanism with
\(
\lambda_\ast=\sqrt{d-1}
\)
up to the threshold \(C_\ast(d)\), followed by calibrated GRR.
Theorem~\ref{thm:ss-opt} shows that within the subset-selection family GRR is again the unique
matched-budget optimizer.

Fourth, Theorem~\ref{thm:equivariant-opt} resolves the full permutation-equivariant problem in the
low-budget regime \(0<C\le C_\ast(d)\).
Every permutation-equivariant channel decomposes into orbit templates.
For the projected unbiased inverse-estimator risk, the signal coefficient and the canonical budget are
additive over orbits, and the orbitwise signal-to-budget slope is maximized only by the singleton GRR
template with \(\lambda_\ast=\sqrt{d-1}\).
Thus low-budget optimality is not merely a feature of the GRR-block class: it is the exact symmetric
optimum, up to conditionally identical refinements of the informative or null symbols.

What remains open is the global symmetric frontier for \(C>C_\ast(d)\).
The low-budget argument is linear in the orbitwise slope and therefore closes the problem only up to the
first tangency point of the GRR curve.
Beyond that threshold, ordered-pair orbits and higher orbit types may still contribute to the true upper
concave envelope, and a complete classification of the high-budget permutation-equivariant frontier remains
to be done.

Several further directions are immediate.
One can ask for sharp nonasymptotic minimax constants beyond the projected unbiased class, for exact
finite-\(n\) comparisons between the privacy curve and the estimation frontier, and for analogous orbit
classifications under additional structural constraints such as bounded message length or sparse outputs.
More broadly, the mechanism-design problem in the shuffle model appears to be governed by a finite-dimensional
orbit geometry rather than by the local-DP extremizers themselves.

\end{document}